\documentclass[letterpaper,twocolumn,10pt]{article}
\usepackage{usenix}

\usepackage{amsmath,amssymb,amsthm}
\usepackage{mathtools}

\usepackage{booktabs}
\usepackage{array}
\usepackage{adjustbox}

\usepackage{algorithm}
\usepackage{algpseudocode}

\usepackage{enumitem}
\usepackage{xspace}

\usepackage{tikz}
\usetikzlibrary{arrows.meta,positioning,fit,backgrounds,calc}

\usepackage{url}
\usepackage{cleveref}

\newtheorem{theorem}{Theorem}
\newtheorem{lemma}{Lemma}
\newtheorem{corollary}{Corollary}
\newtheorem{definition}{Definition}
\newtheorem{proposition}{Proposition}
\newtheorem{remark}{Remark}

\newcommand{\HDKEM}{\textsc{HDKEM}\xspace}
\newcommand{\Gandalf}{\textsc{Gandalf}\xspace}
\newcommand{\FalconRS}{\textsc{FalconRS}\xspace}
\newcommand{\XWing}{\textsc{X-Wing}\xspace}
\newcommand{\Shadowfax}{\textsc{Shadowfax}\xspace}
\newcommand{\RingXKEM}{\textsc{RingXKEM}\xspace}
\newcommand{\Ascon}{\textsc{Ascon}\xspace}

\newcommand{\adv}{\mathsf{Adv}}
\newcommand{\eck}{\textsf{eCK}}
\newcommand{\eps}{\varepsilon}
\newcommand{\epsWI}{\eps_{\mathsf{WI}}}

\newcommand{\Auth}{\mathsf{Auth}}
\newcommand{\Leak}{\mathcal{L}}
\newcommand{\Disc}{\mathcal{D}}
\newcommand{\Simm}{\mathsf{Sim}}

\newcommand{\real}{\mathsf{real}}
\newcommand{\sm}{\mathsf{sim}}

\newcommand{\ksh}{k_{\mathsf{sh}}}
\newcommand{\kconf}{k_{\mathsf{cf}}}
\newcommand{\skPRF}{\mathsf{skPRF}}

\newcommand{\DST}{\mathsf{DST}}
\newcommand{\Hinf}{H_\infty}
\newcommand{\qh}{q_h}

\newcommand{\ID}{\mathsf{ID}}
\newcommand{\len}{\mathsf{len}}
\newcommand{\kid}{\mathsf{kid}}

\newcommand{\rvk}{\mathsf{rvk}}
\newcommand{\rsk}{\mathsf{rsk}}

\begin{document}
%-------------------------------------------------------------------------------

% Do not print a date.
\date{}

% USENIX title format.
% USENIX title format.
\title{\Large \bf
HDKEM: Authentication Allocation for Deniable Hybrid Key Exchange}

\author{
{\rm Kazi Hassan Shakib}\\
Kansas State University\\
kshakib@ksu.edu
\and
{\rm Muhammad Asfand Hafeez}\\
Florida Atlantic University\\
mhafeez2024@fau.edu
\and
{\rm Arslan Munir}\\
Florida Atlantic University\\
arslanm@fau.edu
}

\maketitle

\begin{abstract}
Deniable authenticated key exchange (DAKE) establishes a shared key without leaving transferable evidence that a party participated. We study DAKE protocols that authenticate each fresh session with one ring signature and show that their deniability mechanisms fall into three authentication allocations: sharing an authentication witness, hiding the signing witness under a fixed ring, or varying the verified ring between real and simulated executions. For fixed-ring protocols, witness hiding yields a necessary relation between the ring signature's witness-indistinguishability gap and the protocol's $(\mu,\delta)$ deniability parameters.

We instantiate witness hiding in \HDKEM, a three-slot hybrid that combines ML-KEM-768, X25519, and an optional pairwise information-theoretic source through a split-key PRF. \HDKEM binds all public session values, including a QKD key identifier when present, into authentication and key derivation. Under the stated component models, it achieves eCK key indistinguishability, forward secrecy, and strong local deniability.

We evaluate \HDKEM with \FalconRS and \Gandalf. Across $2{,}000$ primitive operations, median signing/verification times are $9.20/1.02$\,ms for \FalconRS and $18.61/0.97$\,ms for \Gandalf in our reference environment. In source-backed mode, the corresponding TCP-application costs are $8{,}607$\,B and $8{,}555$\,B per handshake. A separate $500$-session server experiment exercises the complete \FalconRS source-backed path and reports a median server time of $16.29$\,ms including local I/O.
\end{abstract}

\section{Introduction}\label{sec:intro}

Deniability prevents a protocol transcript from becoming transferable
evidence that a party participated. This goal becomes harder during
post-quantum migration: ordinary post-quantum signatures provide
authentication but make participation publicly attributable. Ring
signatures offer a natural alternative because either member of a
designated ring can account for an accepted
signature~\cite{Hashimoto22,Hashimoto25,Gajland24,Gajland25,Brendel22,Collins24}.

Existing post-quantum deniable key-exchange constructions obtain
simulation capability in different ways: by sharing an authentication
secret, relying on indistinguishability between signing witnesses, or
changing the ring used in simulation. These choices affect both who can
authenticate and the deniability loss inherited by the protocol, yet
existing analyses do not provide a common structural account of these
trade-offs.

This paper asks: \emph{where must a deniable protocol place the
capability needed to simulate an accepting transcript?} For fresh
sessions authenticated by a single ring signature, we show that this
capability has three allocations: \emph{witness sharing},
\emph{witness hiding}, and \emph{varying ring}. For the fixed-ring class
studied here, witness hiding necessarily satisfies
\[
\delta+(\mu-1)\ge\epsWI,
\]
where $\epsWI$ measures the distinguishability of signatures generated
with the two witnesses. This connects handshake authentication structure
to the $(\mu,\delta)$ deniability notion of Katsumata et
al.~\cite{Katsumata25}.

We instantiate witness hiding in \HDKEM, a hybrid key exchange combining
ML-KEM-768, X25519, and an optional pairwise source such as QKD through a
split-key PRF. \HDKEM authenticates all public session values and achieves
eCK key indistinguishability, forward secrecy, and strong local
deniability under the stated component assumptions. Its fixed-ring design
makes the central trade-off explicit: authentication capability remains
limited to the two endpoints, while deniability inherits the ring
signature's $(\mu,\delta)$ cost.

\subsection{Contributions}\label{sec:contrib}
Our contributions are as follows.

\paragraph{Authentication allocation:}
We define authentication and simulation sets and prove that any simulator producing fresh accepting transcripts must itself possess accepting authentication capability. For the fixed-ring class $\mathcal{F}$, this yields the necessary witness-hiding relation
$\epsWI\le\delta+(\mu-1)b^\star\le\delta+(\mu-1)$.
The resulting taxonomy separates three design choices---witness sharing, witness hiding, and varying ring---and explains why fixed-ring and varying-ring DAKE constructions incur different deniability costs.

\paragraph{Pairwise DAKE deniability:}
We specialize the bundled deniability framework of Katsumata et al.~\cite{Katsumata25} to interactive pairwise exchange. The model distinguishes leakage from the accused party and disclosure by the accuser, local and global denial, and honest-but-curious and malicious insider accusers. The simulator must be executable from the accuser's own secrets and public information.

\paragraph{Robust hybrid handshake:}
We construct \HDKEM from ML-KEM-768, X25519, and an optional pairwise source using a split-key PRF. We characterize the four eCK exposure cases, derive the source-entropy requirements under two PRF analyses, and bind the QKD key identifier into both authentication and key derivation. For QKD-backed deployments, the source interface is reduced to the key-ID-aware CK$^+$ oracle of H\"ovelmanns et al.~\cite{Hovelmanns25}.

\paragraph{Security analysis:}
We prove NIKE-slot soundness, OR-security of the combiner, eCK key indistinguishability, forward secrecy, and strong local deniability, with complementary CryptoVerif and ProVerif models included in the artifact. We also give the source-less two-signature bound, direct and generic multi-session deniability bounds, and an impossibility result for global deniability with the fixed ring.

\paragraph{Concrete evaluation:}
We evaluate the complete protocol stack with \FalconRS and \Gandalf. The measurements show that ring-signature authentication dominates local cryptographic cost, while ML-KEM, X25519, source derivation, and split-key derivation are comparatively small on the reference platform. Exact serialization gives source-backed TCP-application costs of $8{,}607$\,B with \FalconRS and $8{,}555$\,B with \Gandalf; a separate $500$-session server experiment exercises the complete \FalconRS path.

\subsection{Scope and Limitations}\label{sec:scope}
\HDKEM provides local, not global, deniability (Proposition~\ref{prop:noglobal}); \S\ref{sec:aeadmode} gives an optional mode that hides the ring signature from outsiders. The confidentiality and authentication theorems use the proof models stated in \S\ref{sec:sec}. In particular, the published $(\mu,\delta)$ deniability parameters used for \FalconRS and \Gandalf remain classical-ROM, while QROM unforgeability is available for published \Gandalf~\cite{Beckmann26}. The QKD result uses the CK$^+$ key-delivery-oracle model of H\"ovelmanns et al.~\cite{Hovelmanns25}. Implementation details are provided in \S\ref{sec:eval}, and the complementary CryptoVerif and ProVerif models are included in the artifact.

\section{Related Work}\label{sec:related}
\paragraph{Deniable authentication and key exchange:}
Dwork, Naor, and Sahai~\cite{Dwork98} and Di Raimondo et al.~\cite{DiRaimondo06} established deniable authentication and deniable key exchange as simulation problems. Pass~\cite{Pass03} emphasized that denial must be efficiently realizable by the denying party, and Dodis et al.~\cite{Dodis09} studied composability and on-line deniability. We adopt the same implementability principle and ask where simulation capability must reside when a fresh session is authenticated by one ring signature.

\paragraph{Post-quantum deniable handshakes:}
Hashimoto et al.~\cite{Hashimoto22} give a post-quantum, state-leakage-secure X3DH-style construction; Brendel et al.~\cite{Brendel22} study asynchronous deniable key exchange; and Collins et al.~\cite{Collins24} obtain deniability without ring signatures through split-KEM techniques. Hashimoto, Katsumata, and Wiggers~\cite{Hashimoto25} introduce BAKE and \RingXKEM, whose auxiliary unauthenticated ring member provides the varying-ring case in our taxonomy. Katsumata et al.~\cite{Katsumata25} provide the bundled deniability model and the $(\mu,\delta)$ metric used in our concrete analysis. Our allocation theorem is complementary: it characterizes the authentication capability required by fixed-ring single-signature designs.

\paragraph{Ring signatures for deniable AKEM:}
\Gandalf~\cite{Gajland24}, \Shadowfax~\cite{Gajland25}, \FalconRS, and the ring-signature instantiations studied in BAKE show how compact post-quantum ring signatures can support deniable key establishment. \Shadowfax is the closest deniable AKEM combiner to \HDKEM, but it has no independent information-theoretic source slot. \RingXKEM obtains stronger global-style denial by varying the ring; \HDKEM instead studies the fixed-ring witness-hiding point and makes its deniability cost explicit.

\paragraph{Hybrid key establishment and combiners:}
Bindel et al.~\cite{Bindel19} study hybrid KEMs and authenticated key exchange, while Giacon et al.~\cite{Giacon18} show why robust KEM combiners require more than bare XOR. Don, Fehr, and Huang~\cite{DonFehrHuang22} analyze split-key PRFs against quantum multi-oracle adversaries, which supplies the QROM treatment used for our source slot. \XWing~\cite{Connolly24} combines ML-KEM-768 and X25519 and provides the natural two-component confidentiality baseline. Muckle~\cite{Dowling20} gives a framework for quantum-secure hybrid key exchange; KEM--QKD designs authenticated with a shared MAC key fall on the witness-sharing side of our allocation view.

\paragraph{QROM ring signatures and QKD-aware AKE:}
Beckmann and Majenz~\cite{Beckmann26} give QROM reductions for generic AOS and ring-trapdoor/RPSF ring signatures, including strong unforgeability for published \Gandalf. Their result does not establish the classical $(\mu,\delta)$ deniability relation used here, so authentication and deniability retain distinct proof-model qualifications. H\"ovelmanns et al.~\cite{Hovelmanns25} model QKD as a key-delivery oracle in CK$^+$ with key identifiers, leakage, and replacement queries. \HDKEM adopts that interface and authenticates the public key identifier as part of the session transcript.

\section{Preliminaries}\label{sec:prelim}

Let $\lambda$ denote the security parameter. QPT denotes quantum polynomial time,
$\Hinf(X)$ min-entropy, $\qh$ a quantum adversary's random-oracle query budget, and
$q_d$ its decapsulation or distinguishing-query budget.

\subsection{KEM and NIKE}\label{sec:kem}
We use standard KEM and NIKE syntax and defer generic definitions to the appendix/full version~\cite{Katsumata25}. HDKEM instantiates the KEM slot with ML-KEM-768 and
the NIKE slot with X25519. We write $\eps_{\mathsf{mlkem}}$ for the QROM IND-CCA2 advantage and
$\eps_{\mathsf{dh}}$ for the Gap-DH advantage; Table~\ref{tab:prims} records the concrete models.
Crucially, the transcript $\tau$ includes the ML-KEM ciphertext $c$, so changing $c$ changes the
split-key PRF context even if two ciphertexts decapsulate to the same value. Thus
Theorem~\ref{thm:comb} needs no separate ciphertext second-preimage assumption.

\subsection{Comparing Two Distributions}\label{sec:hs}

Deniability compares two distributions: transcripts produced by an honest party
and transcripts an accuser can generate alone. Classical indistinguishability
allows only a negligible additive gap. The following relation additionally
permits a multiplicative gap, recovering the classical notion when $\mu=1$.
We use it in Definitions~\ref{def:rsden}, \ref{def:simset},
and~\ref{def:mudelta}.

\begin{definition}[$(\mu,\delta)$ relation]\label{def:hs}
Let $X$ and $Y$ be distributions, with $\mu\ge 1$ and $\delta\ge 0$. We write
$X\preceq_{\mu,\delta}Y$ if, for every efficient distinguisher $\mathcal{D}$,
\[
  \Pr[\mathcal{D}(X)=0]
  \le
  \mu\,\Pr[\mathcal{D}(Y)=0]+\delta .
\]
\end{definition}

For $\mu=1$, this recovers computational indistinguishability with additive
advantage $\delta$; for $\mu>1$, part of the gap may instead be carried
multiplicatively (Corollaries~\ref{cor:classical}--\ref{cor:relaxed}).
Because the definition quantifies over every $\mathcal{D}$, it also applies to
the complement of any distinguisher and therefore constrains both orientations.
Theorem~\ref{thm:alloc} uses this fact.

\subsection{Ring Signatures}\label{sec:rsprim}

A ring signature lets a signer authenticate $M$ with respect to a ring $R$ of verification keys containing its own:
$\mathsf{RS.Sign}(\rsk,M,R)$.

\begin{definition}[Unforgeability]\label{def:eufcma}
A ring-signature scheme is EUF-CMA-secure with advantage
$\eps_{\mathsf{sig}}$ if no QPT adversary, given the verification keys of a ring
and adaptive signing-oracle access for that ring, produces a valid signature on
a previously unqueried message except with probability
$\eps_{\mathsf{sig}}$.
\end{definition}

\begin{definition}[Anonymity]\label{def:anon}
A ring-signature scheme is $\eps_{\mathsf{anon}}$-anonymous if, given a
signature generated by one of two ring members, no adversary identifies the
signer with advantage greater than $\eps_{\mathsf{anon}}$ over random guessing.
\end{definition}

\begin{definition}[Witness-indistinguishability advantage]\label{def:wi}
Fix a ring $R$ and a signer-independent, efficiently samplable message
distribution $T$. For $b\in\{0,1\}$ let
$\Sigma_b^T=(\tau,\mathsf{RS.Sign}(\rsk_b,\tau,R))$ for $\tau\leftarrow T$.
Then $\epsWI$ is the maximum distinguishing advantage between
$\Sigma_0^T$ and $\Sigma_1^T$. A fixed message is the point-mass case. In
Theorem~\ref{thm:alloc}, $T$ is the fresh transcript distribution induced by
$\Pi$; this avoids silently replacing a pointwise signature comparison by a
comparison on a random protocol transcript.
\end{definition}

\begin{definition}[Deniability {\cite[Def.~9]{Katsumata25}}]\label{def:rsden}
Let $\mu:\mathbb{N}\times\mathbb{N}\rightarrow\mathbb{R}^+$, and let $Q$
bound the number of signing queries. Consider an adversary holding
\emph{every} keypair and querying an oracle that signs under $\rsk_{u_b}$ for a
hidden bit $b$ and ring $\{\rvk_{u_0},\rvk_{u_1}\}$. The scheme is
$(\mu,\delta)$-\emph{deniable} if the two resulting output distributions satisfy
Definition~\ref{def:hs} with parameters $\mu(\lambda,Q)$ and negligible
$\delta$.
\end{definition}

Definition~\ref{def:rsden} is strictly weaker than
Definition~\ref{def:anon}, and this weaker property is sufficient for
deniability: an accused party need only demonstrate that its simulator
\emph{could} have generated the evidence, rather than that the two generation
procedures are equally likely. After the correction of~\cite{Gajland24},
Katsumata et al.~\cite{Katsumata25} report
$\eps_{\mathsf{anon}}=2^{-30}$ for \FalconRS, compared with
$(\mu,\delta)=(1+2^{-27},2^{-57})$ for deniability; \Gandalf achieves
provable deniability at similar parameters. No theorem below relies on
$\eps_{\mathsf{anon}}$; we retain it only to quantify the cost of the
alternative notion.

\subsection{Split-Key Pseudorandom Functions}\label{sec:skprf}

\begin{definition}[Split-key PRF]\label{def:skprf}
A function
\[
  F:\mathcal{S}^n\times\{0,1\}^*\rightarrow\{0,1\}^{512}
\]
is a split-key PRF if the following holds for every slot $i\in[n]$.
A QPT adversary $\mathcal{A}$ may adaptively choose every
$\kappa_j$, $j\neq i$, and a context string $\mathsf{ctx}$, while
$\kappa_i$ is sampled independently from a distribution satisfying
$\Hinf(\kappa_i)\ge\beta$ and is hidden from $\mathcal{A}$. Then
$\mathcal{A}$ distinguishes
\[
  F(\kappa_1,\ldots,\kappa_n;\mathsf{ctx})
\]
from a uniform string of equal length only with negligible advantage.
Each query must use a \emph{fresh} context: no two queries may induce the same
encoded input to $F$ for any value of the hidden slot $\kappa_i$.
\end{definition}

We formulate Definition~\ref{def:skprf} for min-entropy rather than exact
uniformity because the pre-shared-key source of \S\ref{sec:source} need not be
uniform. Context freshness is precisely the condition required by
Don, Fehr, and Huang~\cite{DonFehrHuang22}. \HDKEM satisfies it because
$\mathsf{ctx}$ is the public transcript $\tau$ of Algorithm~\ref{alg:handshake}, which is unique per
session through the ML-KEM ciphertext $c$. This is also condition F3 of
Definition~\ref{def:class}; the same condition makes
Theorem~\ref{thm:alloc} applicable, supplies the $G_2$ hop in
Theorem~\ref{thm:eck}, and discharges the split-key PRF freshness requirement.

\begin{lemma}[Split-key PRF over a min-entropy slot]\label{lem:skprf}
Model $F$ as a quantum random oracle, and suppose some $\kappa_i$ is independent
of all other inputs and the context and satisfies
$\Hinf(\kappa_i)\ge\beta$. Any QPT distinguisher making at most $q$ quantum
queries then has advantage at most
\[
  8(q+1)^2\,2^{-\beta}.
\]
\end{lemma}

\begin{proof}[Proof sketch]
Distinguishing requires locating the hidden $\kappa_i$. Independence makes this
an unstructured search problem, so the generic quantum-search
bound~\cite[Thm.~1]{HRS16} applies; its quadratic dependence on $q$ is
optimal~\cite{Bennett97}. Appendix~\ref{app:fullproofs} gives the full proof.
\end{proof}

\begin{remark}[Two analyses, two entropy requirements]\label{rem:twobounds}
Two published analyses give different bounds for the concrete construction of
\S\ref{sec:skprfinst}; the security theorems of \S\ref{sec:sec} hold under
either, but the resulting entropy requirement differs. Modeling $F$ as a
monolithic quantum random oracle gives Lemma~\ref{lem:skprf}, requiring
$\Hinf(k_3)\ge256$ to obtain
$\eps_{\mathsf{prf}}\le2^{-133}$ at $q=2^{60}$. The multi-oracle analysis of
Don, Fehr, and Huang~\cite{DonFehrHuang22} gives a weaker square-root bound,
requiring $\Hinf(k_3)\ge384$ to keep $\eps_{\mathsf{prf}}$ below ML-KEM's
One-Way-to-Hiding loss of $2^{-98}$. Thus a $32$- or $48$-byte QKD/pre-shared
key suffices under the respective analyses at no additional communication cost.

At $\beta=128$, neither analysis provides the intended guarantee: the
monolithic bound is approximately $2^{-5}$ at $q=2^{60}$, while the
multi-oracle bound is vacuous. Thus merely
requiring ``128-bit security'' is materially weaker than the guarantees used
here. Appendix~\ref{app:prf} derives both bounds.
\end{remark}

\subsection{Concrete Instantiation}\label{sec:skprfinst}

All combiner theorems are stated against Definition~\ref{def:skprf}
(Appendix~\ref{app:staticcmp}). Our concrete construction instantiates $F$ with
SHAKE256:
\begin{equation}\label{eq:skprf}
  \skPRF(\vec{\kappa};\mathsf{ctx})
  =
  \mathsf{SHAKE256}\!\left(
    \DST
    \| \langle\kappa_1\rangle
    \| \cdots
    \| \langle\kappa_n\rangle
    \| \langle\mathsf{ctx}\rangle,\,
    512
  \right),
\end{equation}
where $\langle x\rangle=\len_2(x)\|x$ prefixes $x$ with a two-byte length and
$\DST=\texttt{"HDKEM-v1-skPRF"}$. In the ordinary and source-less modes we
squeeze $512$ bits for the two logical $256$-bit values
$(k_{\mathsf{sh}},k_{\mathsf{cf}})$; participation-hiding mode squeezes
$768$ bits and additionally obtains a logical $256$-bit
$k_{\mathsf{enc}}$. 
Bare XOR does not satisfy Definition~\ref{def:skprf}: an adversary knowing
$n-1$ components can cancel them and reduce the combiner to its remaining
component~\cite{Giacon18}. We use $\eps_{\mathsf{prf}}$ for whichever bound is
selected in Remark~\ref{rem:twobounds}---Lemma~\ref{lem:skprf} in the
monolithic model or~\cite[Cor.~1]{DonFehrHuang22} in the multi-oracle
analysis---rather than as an additional error term. Replacing SHAKE256 with a
standard-model split-key PRF leaves the theorems unchanged, with
$\eps_{\mathsf{prf}}$ replaced by that PRF's distinguishing advantage.

\subsection{Pairwise Source Interface}\label{sec:source}

The third combiner input is delivered by a pairwise source service. A client-side call returns a public handle and a secret,
\[
  (\kid_3,k_3)\leftarrow\mathsf{Src.Get}(C,S),
\]
and the partner obtains the same secret with $k_3\leftarrow\mathsf{Src.GetWithID}(S,C,\kid_3)$. The secret $k_3$ is consumed once and never transmitted; the public $\kid_3$ is included in the signed transcript and split-key-PRF context. For QKD, this interface is the key-ID-aware QKD oracle of H\"ovelmanns et al.~\cite{Hovelmanns25}; for a pre-shared source, $\kid_3$ may be a monotone per-pair index.

We write $\eps_{\mathsf{src}}=\eps_{\mathsf{ks}}=\eps_{\mathsf{qkd}}$
for a QKD source, where $\eps_{\mathsf{qkd}}$ is the trace distance of
the concrete QKD service from the ideal key-delivery interface. For a
pre-shared source, $\eps_{\mathsf{src}}=0$ under the independence and
$\Hinf(k_3)\ge\beta$ premise of Lemma~\ref{lem:skprf}, with its loss
already included in $\eps_{\mathsf{prf}}$. Without the source, the
combiner reduces to $n=2$.

\subsection{Authenticated Encryption}\label{sec:aeadprim}
\Ascon-128a is used as an AEAD, with advantage
$\eps_{\mathsf{aead}}$ against nonce-respecting adversaries, for application
data and key confirmation and, optionally, in the mode of
\S\ref{sec:aeadmode}. The split-key PRF outputs $256$-bit logical values,
whereas Ascon-128a consumes a $128$-bit key; the implementation therefore
sets $K_{\mathsf{cf}}=\operatorname{Trunc}_{128}(k_{\mathsf{cf}})$ and, in
participation-hiding mode,
$K_{\mathsf{enc}}=\operatorname{Trunc}_{128}(k_{\mathsf{enc}})$.
We retain $256$-bit logical subkeys at the combiner interface and truncate
to $128$ bits only for the lightweight \Ascon-128a instantiation, whose
security is correspondingly bounded by its $128$-bit key size.

\begin{table}[t]
\centering
\caption{Primitive assumptions and concrete bounds. Anonymity is comparison-only.}
\label{tab:prims}
\scriptsize
\setlength{\tabcolsep}{2.5pt}
\begin{tabular}{@{}p{1.25cm}p{1.28cm}p{0.92cm}p{2.25cm}@{}}
\toprule
Primitive & Notion & Model & Bound \\
\midrule
ML-KEM-768 & IND-CCA2 & QROM &
$2^{-128}+\sqrt{\qh}\,2^{-128}$ \\
X25519 & Gap-DH & classical &
$2^{-112}$; $\approx1$ vs.\ quantum \\

\Gandalf & SUF-CRA & QROM &
Beckmann--Majenz Thm.~13~\cite{Beckmann26} \\
Ring sig. & $(\mu,\delta)$-deny & cl.\ ROM &
$(1{+}2^{-27},2^{-57})$~\cite{Katsumata25} \\
\emph{RPSF ring sig.} & \emph{ANON} & \emph{QROM} &
\emph{generic full-key-exposure bound~\cite{Beckmann26}} \\
\emph{(ring)} & \emph{QROM deniability} & \emph{QROM} &
\emph{required deniability property unresolved} \\
SHAKE256 & split-key PRF & QROM &
Lem.~\ref{lem:skprf} (mono.) or~\cite[Cor.~1]{DonFehrHuang22} \\
Source & QKD/PSK & inf.-th. &
$\eps_{\mathsf{qkd}}$ / $\Hinf(k_3)\ge\beta$ + independence \\
\Ascon-128a & AEAD & standard &
$\eps_{\mathsf{aead}}$ \\
\bottomrule
\end{tabular}
\end{table}

\section{System and Threat Model}\label{sec:threat}

\paragraph{Deployment:}
\HDKEM runs between a server $S$ and a client $C$. We use these terms rather than
initiator/responder because the party sending the first ephemeral material is determined by which endpoint must expose its ML-KEM public key before encapsulation, as in KEMTLS~\cite{Schwabe21};
roles are therefore deployment-specific rather than tied to connection direction. Endpoints are provisioned with long-term ring keypairs and the pairwise ring
$R=\{\rvk_S,\rvk_C\}$. They may additionally access either a pairwise pre-shared source or a pairwise QKD key-delivery service supplying $(\kid_3,k_3)$. Ephemeral secrets are erased after each session; we assume no tamper-resistant hardware.

\paragraph{Source constraint:}
The source must be pairwise: a group-wide secret known to every endpoint would let one member derive another pair's session keys. In QKD mode the service returns a fresh public key identifier $\kid_3$ with each secret $k_3$; the identifier is not secret and is authenticated as part of $\tau$. Deployments may instead use pair-specific pre-shared secrets or a trusted provisioning service whose master key is \emph{not} exposed to endpoints. Open-enrollment deployments without such provisioning use source-less mode.

\paragraph{Adversary timelines:}
We distinguish two times. At handshake time $t$, the \emph{authentication-time adversary} controls
the network and may be QROM-capable, but has no large-scale quantum computer; ring-signature
unforgeability must hold at $t$. At $t+\Delta$, the
\emph{confidentiality/deniability-time adversary} may be fully quantum-capable. This distinction matters for authentication: under pure harvest-now-decrypt-later, a classical signature could suffice because a later forgery cannot alter a completed handshake. \HDKEM nevertheless uses a post-quantum ring signature because its $(\mu,\delta)$-deniability must remain meaningful at
$t+\Delta$; unlike past unforgeability, deniability could otherwise be broken retrospectively.

\paragraph{Network adversary:}
The adversary $\mathcal{A}$ fully controls the classical channel, is QPT with superposition random-oracle access, and may adaptively reveal session keys, ephemeral state, and long-term keys
subject to the freshness condition below. It may also mount harvest-now-decrypt-later attacks. Access to the source channel requires physical compromise of an endpoint. We treat deniability
adversaries separately in \S\ref{sec:denmodel}.

\subsection{eCK Model and Freshness}\label{sec:eckcases}

We use the extended Canetti--Krawczyk (eCK) model~\cite{LaMacchia07}, with three adaptations. First, the Test query need not be final~\cite{Cremers12}. Second, the adversary is QPT, with QROM
losses accounted for through One-Way-to-Hiding. Third, we add the key-ID-aware pairwise source interface $\mathcal{O}_{\mathsf{Src.Get}}/\mathcal{O}_{\mathsf{Src.GetWithID}}$. For QKD this is instantiated by the CK$^+$ QKD-oracle interface of H\"ovelmanns et al.~\cite{Hovelmanns25}, including QKD-key leakage and replacement queries. The remaining oracles are the standard $\mathcal{O}_{\mathsf{Send}}$, $\mathcal{O}_{\mathsf{Reveal}}$, $\mathcal{O}_{\mathsf{RevealEph}}$, $\mathcal{O}_{\mathsf{Corrupt}}$, and $\mathcal{O}_{\mathsf{RO}}$ oracles.

A session is \emph{fresh} if the adversary has revealed neither its session key nor its partner's,
has not revealed or adversarially replaced the source value referenced by its $\kid_3$, and has not obtained both a party's long-term key and that party's ephemeral state. In QKD mode freshness also requires the matching session to reference the same $\kid_3$; this is the key-ID condition needed to exclude dependent-key attacks.

\paragraph{Freshness cases:}
The eCK freshness condition yields four relevant corruption patterns for participants $A$ and $B$.
Table~\ref{tab:eckcases} shows which \HDKEM confidentiality slots remain hidden in each case.

\begin{table}[t]
\centering
\caption{eCK corruption patterns and surviving \HDKEM confidentiality slots.}
\label{tab:eckcases}
\footnotesize
\begin{tabular}{@{}lccccl@{}}
\toprule
Case & $A$ l-t & $A$ eph. & $B$ l-t & $B$ eph. & Surviving \\
\midrule
C1 & rev. & --- & rev. & --- & $k_1,k_2,k_3$ \\
C2 & rev. & --- & --- & rev. & $k_3$ \\
C3 & --- & rev. & rev. & --- & $k_3$ \\
C4 & --- & rev. & --- & rev. & $k_3$ \\
\bottomrule
\end{tabular}
\end{table}

\paragraph{Why C2--C4 rely on the source:}
All ML-KEM and X25519 secrets are ephemeral, so revealing either endpoint's ephemeral state,
together with the public transcript, recovers both $k_1$ and $k_2$. Hence C2--C4 leave only $k_3$,
while C1 preserves all three slots. This is the cost of a fully ephemeral computational core;
Table~\ref{tab:staticcmp} compares designs retaining a long-term KEM/DH secret.

\paragraph{Out of scope:}
We exclude implementation side channels, fault injection, physical-layer QKD attacks, traffic
analysis, denial of service, and deniability of being provisioned for \HDKEM. Our target is
deniability of participation in a \emph{specific session}, consistent with
prior work~\cite{Katsumata25,Brendel22,Collins24}.

\section{Allocating Authentication Capability}\label{sec:framework}

This section formalizes where deniability places authentication capability.
We define authentication and simulation sets, identify three possible
allocations, and prove an exhaustiveness result for the fixed-ring
single-signature class. 
 The argument depends only on comparing real and simulated transcript distributions; Corollaries~\ref{cor:classical}--\ref{cor:relaxed}
instantiate different comparison regimes. It is otherwise independent of \HDKEM, which
\S\ref{sec:frameworksays} places within the framework.

\subsection{Authentication and Simulation Sets}\label{sec:twosets}

\begin{definition}[Authentication set]\label{def:authset}
$\Auth(\Pi)$ is the set of principals that can, with non-negligible probability, produce a transcript
that an honest verifier accepts as originating from a legitimate endpoint.
\end{definition}

\begin{definition}[Simulation set]\label{def:simset}
$\Simm_P(\Pi,\mu,\delta)$ is the set of principals $Q$ for which an efficient algorithm, runnable by
$Q$ using only its own secrets and public data, produces transcripts related by
Definition~\ref{def:hs} to those genuinely produced by $P$.
\end{definition}

\begin{definition}[Allocation]\label{def:alloc}
An \emph{allocation} for $\Pi$ and $P$ consists of
$(\Auth(\Pi),\Simm_P(\Pi,\mu,\delta))$ together with the secret material held by their members.
\end{definition}

\begin{proposition}[Containment]\label{prop:containment}
For every $\Pi$ and $P$, polynomial $\mu$, and negligible $\delta$,
\[
  \Simm_P(\Pi,\mu,\delta)\subseteq\Auth(\Pi).
\]
\end{proposition}

\begin{proof}[Proof sketch]
Use honest acceptance itself as the distinguisher in Definition~\ref{def:hs}. Any simulator whose
output is sufficiently close to honest transcripts must therefore produce accepting transcripts with
non-negligible probability. Appendix~\ref{app:fullproofs} gives the full proof.
\end{proof}

\subsection{Three Authentication Allocations}\label{sec:allocations}

Proposition~\ref{prop:containment} implies that a simulator must also possess authentication
capability. The following definitions capture three possible authentication
allocations. Theorem~\ref{thm:alloc} shows that, within the fixed-ring
class $\mathcal{F}$, only witness sharing and witness hiding remain;
varying-ring allocations arise by violating F~2.

\begin{definition}[Witness-sharing allocation]\label{def:sharing}
$\Pi$ and $P$ admit a \emph{witness-sharing} allocation with respect to $Q$ if $P$ and $Q$ hold the
same authentication secret. Deniability against $Q$ is then unconditional, but $\Auth(\Pi)$ gains an
additional principal for each such shared secret.
\end{definition}

\begin{definition}[Witness-hiding allocation]\label{def:hiding}
$\Pi$ and $P$ admit a \emph{witness-hiding} allocation with respect to $Q$ if $Q$ authenticates using
a secret distinct from $P$'s and the resulting signature distributions satisfy
Definition~\ref{def:hs}. The authentication set is unchanged; the cost is the
witness-indistinguishability gap $\epsWI$ of Definition~\ref{def:wi}.
\end{definition}

\begin{definition}[Varying-ring allocation]\label{def:varying}
$\Pi$ and $P$ admit a \emph{varying-ring} allocation with respect to $Q$ if the ring verified in
$P$'s executions differs from the ring used in $Q$'s simulations. The authentication set is unchanged,
and no $\epsWI$ bound is required because the signatures are not compared under the same statement.
\end{definition}
\begin{table*}[t]
\centering
\caption{Authentication allocations for fresh single-ring deniable authentication.}
\label{tab:allocations}
\scriptsize
\setlength{\tabcolsep}{4pt}
\begin{tabular}{@{}lllll@{}}
\toprule
Allocation &
Simulator credential &
Authentication capability &
Deniability cost &
Representative case \\
\midrule
Witness sharing &
same witness &
enlarged &
none against witness holder &
symmetric/MAC authentication \\

Witness hiding &
distinct valid witness &
unchanged &
$(\mu,\delta)$ / $\epsWI$ &
\HDKEM \\

Varying ring &
witness for different ring &
unchanged &
avoids same-statement $\epsWI$ &
\RingXKEM auxiliary member \\
\bottomrule
\end{tabular}
\end{table*}
\begin{definition}[Class $\mathcal{F}$]\label{def:class}
A protocol $\Pi$ belongs to $\mathcal{F}$ if:
\begin{description}[leftmargin=*,itemsep=1pt]
\item[F1 (single authenticator).]
Authentication consists of one ring signature $\sigma$ on $\tau$ under ring $R$, verified by
$\mathsf{RS.Verify}(R,\tau,\sigma)$.
\item[F2 (fixed ring).]
$R$ is fixed before the session, known independently of the transcript, and identical across all
accepting executions between a given pair.
\item[F3 (fresh message).]
With all but negligible probability, honest executions use a fresh $\tau$, and honest verifiers reject
previously accepted values.
\end{description}
\end{definition}

F3 is essential rather than technical: without freshness, a principal holding no key could replay an
accepted transcript and thereby enter $\Auth(\Pi)$. In \HDKEM, the ML-KEM ciphertext $c$ provides
this freshness.

\begin{theorem}[Exhaustiveness in $\mathcal{F}$]\label{thm:alloc}
Let $\Pi\in\mathcal{F}$ have ring-signature forgery advantage
$\eps_{\mathsf{sig}}$, and let $P\neq Q$ with
$Q\in\Simm_P(\Pi,\mu,\delta)$. Except for the
$\eps_{\mathsf{sig}}+\mathsf{negl}$ probability of the no-key forgery branch,
one of the following holds:
\begin{enumerate}[label=(\roman*),leftmargin=*,itemsep=1pt]
\item $\Pi$ uses a witness-sharing allocation;
\item $\Pi$ uses a witness-hiding allocation and, for an optimal witness
distinguisher oriented so $a^\star\ge b^\star$,
\[
 \epsWI=a^\star-b^\star
 \le \delta+(\mu-1)b^\star
 \le \delta+(\mu-1),
\]
where $b^\star$ is its probability of outputting $0$ on the simulator's
signing witness.
\end{enumerate}
Condition F2 excludes varying-ring allocations.
\end{theorem}

\begin{proof}[Proof sketch]
Proposition~\ref{prop:containment} gives $Q\in\Auth$. By F3, a simulator
without a ring secret key yields a fresh forgery except with
$\eps_{\mathsf{sig}}+\mathsf{negl}$. If $Q$ holds $P$'s key, (i) holds.
Otherwise, take an optimal distinguisher for Definition~\ref{def:wi}; a
transcript distinguisher simply extracts $(\tau,\sigma)$ and invokes it, so
this witness-to-transcript step is lossless. Writing its real/simulated
probabilities as $a^\star,b^\star$, Definition~\ref{def:hs} gives
$a^\star\le\mu b^\star+\delta$ and hence (ii). Thus
$\eps_{\mathsf{sig}}$ is paid only to exclude the no-key branch, not in the
witness-hiding inequality itself. Appendix~\ref{app:fullproofs} gives the
full reduction.
\end{proof}

\begin{remark}[Tightness of the allocation bound]\label{rem:alloctight}
The first inequality in Theorem~\ref{thm:alloc}(ii) is the exact loss of the
black-box reduction; replacing $b^\star$ by $1$ yields the simpler bound
$\delta+(\mu-1)\ge\epsWI$. Without additional information on
$b^\star$, this relaxation cannot be improved generically. Structure of a
particular AOS/RPSF instantiation could tighten the concrete bound only by
supplying a stronger statement about the two signing distributions (for
example, a bound on $b^\star$ or a direct divergence bound); the allocation
theorem itself assumes no such primitive-specific structure.
\end{remark}

\subsection{Consequences for Deniability Parameters}\label{sec:regimes}

Theorem~\ref{thm:alloc} constrains the additive and multiplicative deniability gaps jointly.

\begin{corollary}[Classical indistinguishability]\label{cor:classical}
Let $\Pi\in\mathcal{F}$ and require classical deniability, i.e., $\mu=1$. Any witness-hiding
allocation then satisfies
\[
  \delta\ge\epsWI .
\]
Hence a scheme whose two signing distributions are distinguishable with advantage $\epsWI$ cannot
achieve a smaller deniability advantage.
\end{corollary}

\begin{proof}[Proof sketch]
Setting $\mu=1$ in Theorem~\ref{thm:alloc}(ii) gives
$\delta\ge\epsWI$. Thus any non-negligible distinguishing gap between
the two signing-witness distributions imposes the same lower bound on
the additive deniability parameter.
\end{proof}

\begin{corollary}[Two-parameter comparison]\label{cor:relaxed}
Let $\Pi\in\mathcal{F}$ use witness hiding and $\mu>1$. Necessarily
$\mu-1\ge\epsWI-\delta$; satisfying this inequality removes the generic
allocation obstruction but does not by itself prove deniability.
\end{corollary}

\begin{proof}[Proof sketch]
Immediate from Theorem~\ref{thm:alloc}(ii). For \FalconRS,
$(\mu,\delta)=(1+2^{-27},2^{-57})$, so the dominant $2^{-27}$ term is carried
multiplicatively rather than additively; Corollary~\ref{cor:multi} records
how this composes across sessions.
\end{proof}

\begin{corollary}[Avoiding both costs requires a varying ring]\label{cor:relgate}
Suppose $Q\in\Simm_P$ with
\[
  \delta+(\mu-1)<\epsWI
\]
and $\Auth(\Pi)$ contains no principal beyond the legitimate endpoints. Then
$\Pi\notin\mathcal{F}$. If F1 and F3 still hold, F2 must fail, yielding a varying-ring allocation.
\end{corollary}

\RingXKEM is the known construction that takes this route: each prekey bundle contains an
unauthenticated ring key whose secret is discarded. No additional principal therefore gains
authentication capability, while a simulator can replace that ring member with one whose secret it
knows. Other session-varying rings can likewise relax F2, but their soundness depends on the concrete ring-signature construction. Corollary~\ref{cor:relgate} establishes only that some departure from a fixed ring is necessary; HDKEM uses the fixed two-member ring.

Theorem~\ref{thm:alloc} applies only to $\mathcal{F}$ and makes no claim about interactive
authentication, trusted setup, or constructions outside the ring-signature setting. Within $\mathcal{F}$, the theorem leaves witness sharing and witness
hiding; varying-ring constructions lie outside $\mathcal{F}$ by violating
F2. Appendix~\ref{app:cells} places representative constructions across
the broader three-way taxonomy.

\subsection{Applying the Framework to \HDKEM}\label{sec:frameworksays}

In its source-backed modes, \HDKEM belongs to $\mathcal{F}$ and uses a
witness-hiding allocation. The client signs with its own long-term ring
key, while the accuser simulates with its distinct ring key; thus
$\Auth(\HDKEM)$ contains only the two endpoints, and deniability inherits
the ring signature's $(\mu,\delta)$ parameters. Source-less mode adds the
server's Step~1 signature and therefore falls outside F1.

Theorem~\ref{thm:alloc} explains two design choices. First, \HDKEM does not derive and share an
authentication key from $k_3$, even though doing so could give $\delta=0$. Such a design is
witness-sharing: it adds an impersonation path and collapses authentication and simulation onto the
same credential, in which case a symmetric authenticator provides the same allocation more cheaply.
A public-key authenticator is useful here precisely because the simulator uses a \emph{different}
secret.

Second, \HDKEM includes no unauthenticated ring member and therefore does not achieve global
deniability. At fixed $\Auth$, Corollary~\ref{cor:relgate} shows that this requires a varying ring;
Proposition~\ref{prop:noglobal} proves the corresponding impossibility for \HDKEM directly.

\section{A Deniability Model for Pairwise DAKE}\label{sec:denmodel}

We specialize the deniability framework of Katsumata et al.~\cite{Katsumata25} from asynchronous
bundled protocols to interactive pairwise DAKE.

\subsection{Entities}\label{sec:entities}

We distinguish four entities. \emph{Accused users} are honest parties
$\mathcal{H}$ that may wish to deny participation in either endpoint role.
\emph{Insider accusers} are corrupted parties $\mathcal{C}$ that executed sessions with an accused
party: an \emph{honest-but-curious} accuser follows the protocol but retains all state, whereas a
\emph{malicious} accuser may deviate arbitrarily. An \emph{outsider accuser} participated in no
session but attempts to establish that two honest parties communicated. Finally, the
\emph{distinguisher} $\mathcal{D}$ renders the verdict, following~\cite{Katsumata25}.

\subsection{Leakage and Disclosure}\label{sec:leakdisc}

\begin{definition}[Leakage and disclosure]\label{def:leak}
Let $\mathsf{leak},\mathsf{disc}\in\{\mathsf{low},\mathsf{med},\mathsf{high}\}$.
For an accused user $u\in\mathcal{H}$ with ring secret $\rsk_u$, pairwise source secrets
$k_3^{(u,\cdot)}$, and session state $\mathsf{st}_u$, the leakage function reveals:
nothing at $\mathsf{low}$; $\rsk_u$ at $\mathsf{med}$; and
$(\rsk_u,k_3^{(u,\cdot)},\mathsf{st}_u)$ at $\mathsf{high}$.

For an accuser $v\in\mathcal{C}$, the disclosure function reveals:
$\rsk_v$ at $\mathsf{low}$; additionally $v$'s source secrets and current session state at
$\mathsf{med}$; and additionally all consumed randomness at $\mathsf{high}$.
\end{definition}

\subsection{Scopes and Metric}\label{sec:scopes}

\emph{Local deniability} allows an accused party to deny a session with the accuser.
\emph{Global deniability} additionally allows two accused parties to deny the same session
simultaneously, thereby covering outsider accusers.

\begin{definition}[$(\mu,\delta)$-deniability]\label{def:mudelta}
Let
$\mathsf{Game}^{\mathsf{ATK}}_{\mathcal{D},\mathcal{H},\Leak,\Disc}
(1^\lambda,\mathsf{mode})$,
for $\mathsf{ATK}\in\{\mathsf{local},\mathsf{global}\}$ and
$\mathsf{mode}\in\{\real,\sm\}$, denote Algorithm~\ref{alg:dengame}
(Appendix~\ref{app:dengame}). Parties are keyed and pairwise sources provisioned.
The distinguisher $\mathcal{D}$ interacts with an oracle that either executes the real protocol or,
in $\sm$-mode, invokes a simulator $\Simm$ for the relevant party. It then receives the prescribed
leakage and disclosure and outputs a verdict.

A protocol is $(\mu,\delta)$-$\mathsf{ATK}$-deniable with respect to $\Leak$ and $\Disc$ if there
exists a simulator $\Simm$, runnable by the accuser using only its own secrets and public data, such
that the real and simulated games satisfy Definition~\ref{def:hs}. That is, for every efficient

\[
\begin{aligned}
\Pr[\mathsf{Game}^{\mathsf{ATK}}(1^\lambda,\real)=0]
&\le \mu(\lambda,Q)\,
\Pr[\mathsf{Game}^{\mathsf{ATK}}(1^\lambda,\sm) \\
&\quad =0] + \delta(\lambda).
\end{aligned}
\]

where $Q$ bounds $\mathcal{D}$'s queries.

\emph{Strong deniability} replaces the real-world accuser by an arbitrary efficient adversary and
allows the simulator to depend on that adversary.
\end{definition}

\begin{remark}[Simulator inputs]\label{rem:simargs}
Algorithm~\ref{alg:dengame} gives the simulator the \emph{accuser's} ring secret and pairwise source
secret, because the accuser itself must be able to run the simulation. This implementability
requirement is precisely Pass's constraint~\cite{Pass03}. A simulator requiring the accused party's
secret key would therefore not satisfy Definition~\ref{def:mudelta}, regardless of how closely its
output matches the real distribution.
\end{remark}

\section{The \HDKEM Protocol}\label{sec:proto}

\subsection{Ring-Signature Instantiation}\label{sec:rsinst}
\HDKEM accepts any EUF-CMA-secure, $(\mu,\delta)$-deniable ring signature satisfying the abstract interface used by Theorems~\ref{thm:eck} and~\ref{thm:deny}. We use two concrete instantiations in the evaluation.

\paragraph{\FalconRS:}
We instantiate the two-member \FalconRS construction of Katsumata et al.~\cite{Katsumata25}, using Falcon-512 parameters, an 897-byte verification key, and a 1288-byte two-member ring
signature.

\paragraph{\Gandalf:}
Our second instantiation follows the \Gandalf signing and verification
construction of Gajland et al.~\cite{Gajland24}, using its published
896-byte verification key and 1236-byte two-member signature format.
The implementation-specific choices used in our measurements are
described in \S\ref{sec:setup}.

At the primitive level, published \Gandalf has a QROM unforgeability reduction~\cite{Beckmann26}; the concrete deniability parameters used here remain classical-ROM (Remark~\ref{rem:quantumdeny}). Ring-signature and split-key-PRF hashes use distinct SHAKE256 domain tags.

\subsection{Protocol Specification}\label{sec:protospec}
 
HDKEM runs between a server $S$ and a client $C$, each holding a
long-term ring keypair and the pairwise ring $R = \{\rvk_S, \rvk_C\}$
. All KEM and NIKE material is generated
per session and erased after key derivation. The handshake consists
of three sub-protocols run in order: the responder hello
(Algorithm~\ref{alg:hello}, run by $S$), the initiator handshake
(Algorithm~\ref{alg:handshake}, run by $C$), and responder
verification and key confirmation (Algorithm~\ref{alg:verify}, run
by $S$). Figure~\ref{fig:flow} gives the corresponding message-flow
diagram; the algorithms below give the same protocol at the level of
individual operations, in the style of~\cite{Katsumata25}.
 
Throughout, $\langle \cdot \rangle$ denotes the injective length-prefixed
encoding of \S\ref{sec:skprfinst}, and \textbf{require} aborts the
executing algorithm (returning $\bot$) if the guarded condition is
false.
 \paragraph{Canonical transcript encoding:}
We define the canonical transcript encoding as
\[
\begin{aligned}
\mathsf{Tr}(&pk_1^S,c,pk_2^C,pk_2^S,
             \ID_C,\ID_S,\mathsf{mode},\kid_3) \\
&=
pk_1^S\|c\|pk_2^C\|pk_2^S \\
&\quad\|\langle\ID_C\rangle
\|\langle\ID_S\rangle \\
&\quad\|\mathsf{mode}
\|\langle\kid_3\rangle .
\end{aligned}
\]
The same byte string is used for signing, verification, key derivation, and key confirmation.
\paragraph{Setup:} At commissioning, each endpoint generates a long-term
ring keypair, both install $R=\{\rvk_S,\rvk_C\}$ out of band, and deployments
optionally provision the pairwise source service. These values are static;
all KEM/NIKE material below is per-session.
 
\begin{algorithm}[t]
\caption{Responder hello. Step~1 is signed only in source-less mode.}
\label{alg:hello}
\begin{algorithmic}[1]
\Function{HDKEM.Hello}{$\rsk_S,R,\mathsf{mode},\ID_S$}
  \State $(pk_1^S,sk_1^S)\gets\mathsf{MLKEM.Gen}()$;
         $(pk_2^S,sk_2^S)\gets\mathsf{X25519.Gen}()$
  \If{$\mathsf{mode}=\textsc{source-less}$}
    \State $\mathsf{hdr}\gets
    pk_1^S\|pk_2^S\|\langle\ID_S\rangle\|\mathsf{mode}$
    \State $\sigma_S\gets\mathsf{RS.Sign}(\rsk_S,\mathsf{hdr},R)$
  \Else
    \State $\sigma_S\gets\bot$
  \EndIf
  \State $\mathsf{st}_S\gets(pk_1^S,sk_1^S,pk_2^S,sk_2^S)$
  \State \Return
  $(\rho_1:=(pk_1^S,pk_2^S,\ID_S,\mathsf{mode},\sigma_S),
    \mathsf{st}_S)$
\EndFunction
\end{algorithmic}
\end{algorithm}

\begin{algorithm}[t]
\caption{Initiator handshake; source-backed modes are
$\{\textsc{src},\textsc{hide}\}$.}
\label{alg:handshake}
\begin{algorithmic}[1]
\Function{HDKEM.Handshake}{$\rsk_C,R,\ID_C,\rho_1$}
  \State Parse $(pk_1^S,pk_2^S,\ID_S,\mathsf{mode},\sigma_S)\gets\rho_1$
  \If{$\mathsf{mode}=\textsc{source-less}$}
    \State $\mathsf{hdr}\gets
    pk_1^S\|pk_2^S\|\langle\ID_S\rangle\|\mathsf{mode}$
    \State \textbf{require}
    $\mathsf{RS.Verify}(R,\mathsf{hdr},\sigma_S)=1$
  \EndIf

  \State $(pk_2^C,sk_2^C)\gets\mathsf{X25519.Gen}()$;
         $(c,k_1)\gets\mathsf{MLKEM.Encaps}(pk_1^S)$
  \State $k_2\gets\mathsf{DH}(sk_2^C,pk_2^S)$

  \If{$\mathsf{mode}\in\{\textsc{src},\textsc{hide}\}$}
    \State $(\kid_3,k_3)\gets\mathsf{Src.Get}(\ID_C,\ID_S)$
  \Else
    \State $(\kid_3,k_3)\gets(\bot,\bot)$
  \EndIf

  \State $\tau\gets
  \mathsf{Tr}(pk_1^S,c,pk_2^C,pk_2^S,
              \ID_C,\ID_S,\mathsf{mode},\kid_3)$

  \If{$\mathsf{mode}\in\{\textsc{src},\textsc{hide}\}$}
    \State $(\ksh,\kconf[,k_{\mathsf{enc}}])
    \gets\skPRF(k_1,k_2,k_3;\tau)$
    \Comment{$k_{\mathsf{enc}}$ iff \textsc{hide}}
  \Else
    \State $(\ksh,\kconf)\gets\skPRF(k_1,k_2;\tau)$
  \EndIf

  \State $\sigma_C\gets\mathsf{RS.Sign}(\rsk_C,\tau,R)$
  \If{$\mathsf{mode}=\textsc{hide}$}
    \State $\nu_\sigma\gets$ fresh nonce;
    $\mathsf{auth}_C\gets
    (\nu_\sigma,\mathsf{AEAD.Enc}
    (K_{\mathsf{enc}},\nu_\sigma,\sigma_C))$
  \Else
    \State $\mathsf{auth}_C\gets\sigma_C$
  \EndIf

  \State \Return
  $(\rho_2:=(c,pk_2^C,\ID_C,\kid_3,\mathsf{auth}_C),
    \ksh,\kconf,\tau)$
\EndFunction
\end{algorithmic}
\end{algorithm}

\begin{algorithm}[t]
\caption{Responder verification and key derivation.}
\label{alg:verify}
\begin{algorithmic}[1]
\Function{HDKEM.Verify}{$\mathsf{st}_S,\rho_2$}
  \State Parse $(c,pk_2^C,\ID_C',\kid_3,\mathsf{auth}_C)\gets\rho_2$;
         \textbf{require} $\ID_C'=\ID_C$
  \State $k_1\gets\mathsf{MLKEM.Decaps}(sk_1^S,c)$;
         $k_2\gets\mathsf{DH}(sk_2^S,pk_2^C)$
  \State $\tau\gets
  \mathsf{Tr}(pk_1^S,c,pk_2^C,pk_2^S,
              \ID_C,\ID_S,\mathsf{mode},\kid_3)$

  \If{$\mathsf{mode}\in\{\textsc{src},\textsc{hide}\}$}
    \State $k_3\gets
    \mathsf{Src.GetWithID}(\ID_S,\ID_C,\kid_3)$
    \State $(\ksh,\kconf[,k_{\mathsf{enc}}])
    \gets\skPRF(k_1,k_2,k_3;\tau)$
    \Comment{$k_{\mathsf{enc}}$ iff \textsc{hide}}
  \Else
    \State $(\ksh,\kconf)\gets\skPRF(k_1,k_2;\tau)$
  \EndIf

  \If{$\mathsf{mode}=\textsc{hide}$}
    \State Parse $(\nu_\sigma,C_\sigma)\gets\mathsf{auth}_C$
    \State $\sigma_C\gets
    \mathsf{AEAD.Dec}(K_{\mathsf{enc}},\nu_\sigma,C_\sigma)$;
    \textbf{require} $\sigma_C\neq\bot$
  \Else
    \State $\sigma_C\gets\mathsf{auth}_C$
  \EndIf

  \State \textbf{require} $\mathsf{RS.Verify}(R,\tau,\sigma_C)=1$
  \State \Return $(\ksh,\kconf,\tau)$
\EndFunction
\end{algorithmic}
\end{algorithm}

\paragraph{Key confirmation:} After verification, $S$ and $C$ exchange
Finished records $(\nu_R,\mathsf{AEAD.Enc}(K_{\mathsf{cf}},\nu_R,
\mathsf{label}_R\|\tau))$ for roles $R\in\{S,C\}$ using fresh nonces. Each
side accepts only if decryption returns the expected role label and the
identical $\tau$; otherwise it aborts without outputting $\ksh$. Thus a
source mismatch is detected before application keys are used.
 
\begin{remark}[Public source handle, secret source key]\label{rem:excl-alg}
The public source identifier $\kid_3$ is deliberately included in the signed transcript $\tau$, while the secret $k_3$ is not. A ring signature hides the signer, not the message; putting $k_3$ on the signed wire transcript would disclose the source secret. Binding $\kid_3$, however, is necessary: it identifies the exact source entry used by the session and prevents cross-session QKD-key substitution. The local-deniability simulator knows the accuser's matching pairwise source state and can therefore reconstruct the same public handle and derived keys without learning any secret of the accused party.
\end{remark}
 
\begin{remark}[Source-less fallback]\label{rem:srcless-alg}
When $\mathsf{mode} = \textsc{source-less}$, Algorithm~\ref{alg:hello}
additionally signs $\rho_1$ under $\rsk_S$, since $k_3$ is unavailable
to provide server authentication through the Finished exchange.
This is the two-signature configuration analyzed in
Remark~\ref{rem:srcless} and costs the deniability bound
$(\mu_{\mathsf{RS}}^2, \delta_{\mathsf{RS}}(1+\mu_{\mathsf{RS}}))$
rather than $(\mu_{\mathsf{RS}}, \delta_{\mathsf{RS}})$.
\end{remark}
 
\paragraph{End-to-end execution:} A complete session runs
$\textsc{HDKEM.Hello} \to \textsc{HDKEM.Handshake} \to
\textsc{HDKEM.Verify} \to \textsc{HDKEM.Finish}$, invoked twice (once
per role) for the last function. Both parties terminate holding the
same $k_{\mathsf{sh}}$ if and only if all \textbf{require} guards
pass; any failure aborts with $\bot$ and no session key is output,
matching the freshness and injective-agreement properties used in
Theorem~\ref{thm:eck}.

\paragraph{Transcript binding:}
Because Step~1 is unauthenticated in \textsc{src} mode, $\tau$ binds $pk_1^S$ itself rather than only
$c$. Replacing $pk_1^S$ therefore changes the signed transcript and is detected at signature verification. Binding both identities, the ML-KEM ciphertext, and the mode also makes each accepted transcript session-specific. This serves three purposes: it prevents replay in the $G_2$ hop of Theorem~\ref{thm:eck}; supplies F3 of Definition~\ref{def:class}, enabling
Theorem~\ref{thm:alloc}; and satisfies the context-freshness requirement of
Definition~\ref{def:skprf}, required by~\cite[Cor.~1]{DonFehrHuang22}.
Keeping the secret $k_3$ outside $\tau$ preserves source secrecy, while including the public $\kid_3$ binds the exact source instance into both authentication and key derivation. A key-ID substitution therefore changes the signed statement, and inconsistent source secrets yield different $\kconf$ values and fail during key confirmation before application keys are used.

\paragraph{Modes and an optional fourth slot:}
A deployment with a long-term device anchor $k_0$ may derive
\[
  \ksh\|\kconf
  \gets
  \skPRF(k_0,k_1,k_2,k_3;\tau),
\]
adding another OR-confidentiality slot without changing the deniability mechanism. Since
$\mathcal{O}_{\mathsf{Corrupt}}$ reveals $k_0$ together with the long-term ring keys, however, this
does not improve the C2/C3 eCK cases.

The selected $\mathsf{mode}$ is bound into both $\tau$ and $\sigma$, so downgrade requires a signature forgery. In source-less mode, $S$ additionally signs Step~1 because $k_3$ is unavailable for server authentication. Confidentiality then falls back to the two computational slots and eCK
security to C1, while deniability remains ring-signature-based but incurs the two-signature cost of Remark~\ref{rem:srcless}.

\subsection{Optional Participation Hiding and QKD Source}\label{sec:aeadmode}
\label{sec:qkd}

In \textsc{hide}, a source-backed participation-hiding mode, $\skPRF$ additionally derives $k_{\mathsf{enc}}$ and the client encrypts $\sigma_C$ under $K_{\mathsf{enc}}$, following \RingXKEM~\cite{Hashimoto25}. The honest client chooses a fresh AEAD nonce and sends it with the ciphertext. Because each session uses a fresh $K_{\mathsf{enc}}$ for one such encryption, nonce-respecting AEAD security suffices. Decryption or signature-verification failure leads to the same protocol abort. This mode hides the ring signature from outsiders; it does not strengthen deniability against the session peer, which also derives $K_{\mathsf{enc}}$. Traffic-analysis and implementation timing leakage are outside our model.

When $k_3$ is supplied by QKD, \HDKEM consumes an authenticated QKD key-delivery service rather than authenticating QKD reconciliation itself. The service exposes a public identifier $\kid_3$ and the corresponding secret through the interface of H\"ovelmanns et al.~\cite{Hovelmanns25}; $\kid_3$ is signed and included in the split-key-PRF context. Appendix~\ref{app:qkd} gives the CK$^+$ reduction. A provisioned pairwise source satisfies the same \HDKEM interface and is used for the performance experiments in \S\ref{sec:eval}.

\section{Security Analysis}\label{sec:sec}

\paragraph{Proof models:}
The claims use component-specific proof models. ML-KEM and SHAKE256 are analyzed in the QROM, while the X25519 slot relies on classical Gap-DH; against a quantum attacker, confidentiality therefore requires a surviving post-quantum or information-theoretic slot. Authentication inherits the unforgeability model of the selected ring signature. For published \Gandalf, Beckmann--Majenz~\cite{Beckmann26} give a QROM strong-unforgeability reduction. Deniability reduces to the ring signature's $(\mu,\delta)$ guarantee; the concrete parameters used in this paper remain classical-ROM (Remark~\ref{rem:quantumdeny}).
\paragraph{Formal protocol analysis:}
The artifact includes complementary CryptoVerif and ProVerif models for
protocol-level secrecy, authentication, agreement, source handling, and
key confirmation; primitive assumptions and QKD physical security remain external.

\subsection{NIKE-Slot Security}\label{sec:nike}

\begin{theorem}[NIKE-slot soundness]\label{thm:nike}
Suppose both X25519 public values are included in $\tau$ and hence bound into both the signature and
key derivation. If the ring signature is EUF-CMA-secure with advantage
$\eps_{\mathsf{sig}}$, then both parties derive the same $k_2$, and under Gap-DH,
\[
  \adv^{\mathsf{nike}}_{\HDKEM}(\mathcal{A})
  \le \eps_{\mathsf{dh}}+\eps_{\mathsf{sig}}.
\]
\end{theorem}

\begin{proof}[Proof sketch]
Substitution changes $\tau$ and therefore requires a forgery; otherwise embed Gap-DH in the two X25519 values. Appendices~\ref{app:fullproofs} and~\ref{app:nike} give the reduction.
\end{proof}

\subsection{Combiner OR-Security}\label{sec:comb}

\begin{theorem}[Robust combiner]\label{thm:comb}
Let
\[
  (\ksh,\kconf)\gets\skPRF(k_1,k_2,k_3;\tau),
\]
where $\tau$ contains $c$. For each surviving secure slot
$i\in I_{\mathsf{sec}}\subseteq\{1,2,3\}$, let
\[
\eps_1=\eps_{\mathsf{mlkem}},\qquad
\eps_2=\eps_{\mathsf{dh}},\qquad
\eps_3=\eps_{\mathsf{src}}.
\]
Then
\[
  \adv^{\mathsf{comb}}_{\HDKEM}(\mathcal{A})
  \le
  \min_{i\in I_{\mathsf{sec}}}\eps_i+\eps_{\mathsf{prf}}.
\]
\end{theorem}

\begin{proof}[Proof sketch]
Reduce to any surviving slot at cost $\eps_i$, then apply split-key PRF
security at cost $\eps_{\mathsf{prf}}$. Because $c$ is included in $\tau$, any
$c'\neq c$ induces a distinct PRF context even if both ciphertexts decapsulate to
the same $k_1$; no ciphertext second-preimage assumption is required.
Appendix~\ref{app:combproof} gives the full reduction.
\end{proof}

\begin{remark}[Meaning of slot compromise]\label{rem:compromise}
Saying that the session key tolerates compromise of any two slots means that revealing any two of
$k_1,k_2,k_3$ through $\mathcal{O}_{\mathsf{RevealEph}}$,
$\mathcal{O}_{\mathsf{Corrupt}}$, or $\mathcal{O}_{\mathsf{Src}}$, before the freshness check, still
leaves the third slot secure under its stated premise. The claim does not cover simultaneous compromise of all three.
\end{remark}

\subsection{eCK Key Indistinguishability}\label{sec:eck}

\begin{theorem}[eCK security]\label{thm:eck}
For every fresh session,
\[
  \adv^{\eck}_{\HDKEM}(\mathcal{A})
  \le
  \eps_{\mathsf{sig}}
  +\eps_{\mathsf{slot}}
  +2\eps_{\mathsf{prf}}
  +\eps_{\mathsf{aead}},
\]
where
\[
\eps_{\mathsf{slot}}
=
\begin{cases}
\min\{\eps_{\mathsf{mlkem}},\eps_{\mathsf{dh}},\eps_{\mathsf{src}}\},
    & \text{C1},\\[1mm]
\eps_{\mathsf{src}},
    & \text{C2--C4}.
\end{cases}
\]
Here $\eps_{\mathsf{sig}}$ is measured in the ring-signature unforgeability
model, $\eps_{\mathsf{prf}}$ is as fixed in Remark~\ref{rem:twobounds},
$\eps_{\mathsf{src}}$ is defined in \S\ref{sec:source}, and
$\eps_{\mathsf{dh}}$ denotes the applicable classical Gap-DH bound for X25519.
\end{theorem}

\begin{proof}[Proof sketch]
Games eliminate signature/confirmation failures, enforce signed $\kid_3$
consistency, reduce to one surviving slot, and finally idealize the PRF output.
C1 uses the minimum of the three slot bounds; C2--C4 use
$\eps_{\mathsf{src}}$. Appendix~\ref{app:fullproofs} gives the full sequence.
\end{proof}

\paragraph{QKD instantiation:} Appendix~\ref{app:qkd} gives the CK$^+$ key-ID-aware source reduction~\cite{Hovelmanns25}; signing $\kid_3$ excludes the corresponding substitution attack.

\paragraph{Concrete bound:}
At $\qh=2^{60}$, ML-KEM contributes about $2^{-98}$. In C2--C4 only $k_3$ survives; under the multi-oracle PRF analysis, increasing the source from $\beta=256$ to $384$ improves the PRF term from about $2^{-45.5}$ to $2^{-109.5}$, restoring the ML-KEM ceiling. Remark~\ref{rem:twobounds} gives both derivations; no anonymity term enters the bound.

\subsection{Forward Secrecy}\label{sec:fs}

\begin{theorem}[Forward secrecy]\label{thm:fs}
After session completion and erasure of the ephemeral ML-KEM and X25519 secret
keys, subsequent compromise of both parties' long-term ring-signing keys, but not
the pairwise source $k_3$, gives advantage at most
\[
  \min\{\eps_{\mathsf{mlkem}},\eps_{\mathsf{dh}},\eps_{\mathsf{src}}\}
  +\eps_{\mathsf{prf}}
\]
against the completed session key.
\end{theorem}

\begin{proof}[Proof sketch]
The long-term ring-signing keys do not enter $k_1$, $k_2$, $k_3$, or the
split-key PRF input. After ephemeral erasure, all three confidentiality slots
remain available; applying the robust-combiner reduction gives the bound.
Appendix~\ref{app:fullproofs} gives the reduction.
\end{proof}

\subsection{Deniability}\label{sec:denresults}

\begin{theorem}[Strong local deniability]\label{thm:deny}
If the ring signature is
$(\mu_{\mathsf{RS}},\delta_{\mathsf{RS}})$-deniable, then \HDKEM in \textsc{src} mode is
$(\mu_{\mathsf{RS}},\delta_{\mathsf{RS}})$-strongly locally deniable under
Definition~\ref{def:mudelta}, even for
$\mathsf{leak}=\mathsf{high}$ and $\mathsf{disc}=\mathsf{high}$ and against malicious accusers.
\end{theorem}

\begin{proof}[Proof sketch]
Fix accused party $P$ and accuser $Q$, which holds $\rsk_Q$ and the pairwise source
$k_3^{(P,Q)}$. In \textsc{src} mode only the client signs, so the nontrivial case is $P=C$.
The simulator samples fresh client X25519 material, encapsulates to the public server ML-KEM key,
forms $\tau$ from public values, derives $\ksh\|\kconf$ using $k_3^{(P,Q)}$, and signs with
$\rsk_Q$. Thus it uses only the accuser's secrets and public data.

Real and simulated transcripts differ only in the signing witness: ephemeral values are sampled
identically, and $\tau,\ksh,\kconf$ are deterministic from values available to the simulator.
Replacing $\rsk_P$ by $\rsk_Q$ is therefore exactly the ring-signature deniability game. High
leakage and disclosure add no extra loss because the underlying game already exposes all keypairs
and the remaining state is identically distributed. Full proof appears in
Appendix~\ref{app:fullproofs}.
\end{proof}

\begin{remark}[Quantum status of deniability]\label{rem:quantumdeny}
Theorem~\ref{thm:deny} inherits the ring signature's $(\mu,\delta)$ guarantee. The concrete \FalconRS/\Gandalf parameters used here are classical-ROM~\cite{Katsumata25}; the QROM results of Beckmann--Majenz~\cite{Beckmann26} establish unforgeability and anonymity properties, but not the deniability relation required by Theorem~\ref{thm:deny}.
\end{remark}

\begin{remark}[Source-less mode]\label{rem:srcless}
Source-less mode contains two signatures. Replacing them sequentially applies the deniability hop
twice, yielding
\[
  \left(
    \mu_{\mathsf{RS}}^2,\,
    \delta_{\mathsf{RS}}(1+\mu_{\mathsf{RS}})
  \right).
\]
\end{remark}

\begin{corollary}[Multi-session deniability]\label{cor:multi}
Let $(\mu_1,\delta_1)$ denote the ring signature's one-query deniability
parameters and $(\mu_\ell,\delta_\ell)$ its Definition~\ref{def:rsden}
parameters for $Q=\ell$ queries. For $\ell$ sessions under one ring key,
\HDKEM satisfies both the direct bound $(\mu_\ell,\delta_\ell)$ and the
generic hybrid bound
\[
 \left(\mu_1^\ell,
 \delta_1\sum_{j=0}^{\ell-1}\mu_1^j\right).
\]
\end{corollary}

\begin{proof}[Proof sketch]
For the direct bound, one ring-signature deniability adversary answers all
$\ell$ signing challenges while simulating every other HDKEM component, so
there is no session-by-session loss beyond the primitive's own $Q=\ell$
parameters. Independently, replacing one session at a time and unrolling
$a_{i-1}\le\mu_1a_i+\delta_1$ gives the displayed geometric bound. The
hybrid argument does not require independent transcripts: public metadata
may be shared or correlated provided it is generated identically in the two
adjacent hybrids and the ring-signature guarantee remains valid with that
auxiliary information. Correlation through reused signer randomness or
witness-dependent metadata is not covered and must be charged at the
primitive/implementation level. Appendix~\ref{app:fullproofs} gives both
reductions.
\end{proof}

\begin{remark}[Tightness and rotation guidance]\label{rem:degrade}
The geometric expression is a conservative composition upper bound, not a
matching lower bound, and we have no evidence that typical distinguishers
attain it. Whenever a concrete $Q=\ell$ ring-signature bound is available,
the direct $(\mu_\ell,\delta_\ell)$ reduction should be used instead. With
only the reported per-session values $(1+2^{-27},2^{-57})$, the geometric
bound keeps multiplicative loss below $1+2^{-7}$ for about $2^{20}$ sessions
and becomes vacuous by about $2^{40}$; this cadence is proven-safe, not
proven-optimal.
\end{remark}

\begin{proposition}[No global deniability]\label{prop:noglobal}
\HDKEM is not globally deniable for polynomial $\mu$ and negligible $\delta$.
\end{proposition}

\begin{proof}[Proof sketch]
A global simulator holding neither endpoint's secret key can produce a valid signature for $R=\{\rvk_S,\rvk_C\}$ only by an EUF-CMA forgery. Simulated acceptance is therefore negligible, while real acceptance is overwhelming.
\end{proof}

By Corollary~\ref{cor:relgate}, this limitation is structural rather than accidental: at fixed $\Auth$, global deniability requires a varying ring. \S\ref{sec:scope} discusses its cost.

\begin{table}[t]
\centering
\caption{Single-session deniability guarantees in \textsc{src} mode
using the reported \FalconRS parameters
$(1{+}2^{-27},2^{-57})$. Source-less mode is given by
Remark~\ref{rem:srcless}; multi-session degradation by
Corollary~\ref{cor:multi}.}

\label{tab:denresults}
\footnotesize
\begin{tabular}{@{}llll@{}}
\toprule
Scope & Accuser & $(\mathsf{leak},\mathsf{disc})$ & $(\mu,\delta)$ \\
\midrule
local & honest-but-curious & (high, high) & $(1{+}2^{-27},2^{-57})$ \\
local & malicious & (high, high) & $(1{+}2^{-27},2^{-57})$ \\
global & any & --- & none (Prop.~\ref{prop:noglobal}) \\
\bottomrule
\end{tabular}
\end{table}

\section{Implementation and Evaluation}\label{sec:eval}

We evaluate the computational and communication costs of \HDKEM, with complete end-to-end measurements for the source-backed configuration.

\subsection{Implementation and Methodology}\label{sec:setup}
\paragraph{Protocol implementation:}
The client, server, and microbenchmark share one implementation of ML-KEM-768, X25519, SHAKE256, \Ascon-128a, transcript serialization, the pairwise source interface, and ring-signature dispatch. The implementation uses the canonical $\mathsf{Tr}$ encoding of \S\ref{sec:protospec} unchanged for signing, verification, and the split-key PRF. The source secret $k_3$ is never serialized, and ring verification keys are provisioned out of band.

For \FalconRS, we implement the two-member \FalconRS construction of Katsumata et al.~\cite{Katsumata25}. The \Gandalf path follows the signing and verification construction of Gajland et al.~\cite{Gajland24}
and its published wire dimensions. In the measured code, key generation and signer sampling use Falcon routines from PQClean rather than the Antrag/Mitaka components of the published Gandalf implementation; nonsigner sampling and the fixed-width 606-byte vector encoding are implemented locally. We therefore report these measurements as timings for our Gandalf-compatible evaluation path rather than as measurements of the published implementation; \S\ref{sec:rsinst} states the corresponding protocol-level abstraction.

\paragraph{Methodology:}
Primitive benchmarks use $2{,}000$ operations and report median, p05, and p95. We additionally run $500$ successful source-backed \FalconRS server sessions covering ephemeral generation, ring verification, ML-KEM decapsulation, X25519, source lookup, split-key derivation, Finished processing, framing, and local socket I/O. The benchmark source derives a 32-byte $k_3$ from a pairwise master secret and fresh 16-byte $\kid_3$ to exercise the source interface; it does not model QKD generation or physical-layer performance.
\paragraph{Benchmark environment:}
Measurements were collected in a VirtualBox VM with four vCPUs and
8 GB RAM, running Ubuntu~22.04.5 LTS
(kernel~6.8.0-136-generic) on an AMD Ryzen~5~8640HS host. We compiled
with GCC~12.3.0 using \texttt{-O2 -g -std=c11}; binaries link against
liboqs~9, OpenSSL~3, and libsodium~26. All reported experiments use this
same VM and build configuration.

\begin{table}[t]
\centering
\caption{Primitive microbenchmarks
($2{,}000$ operations, 32-byte source), in ms;
\Gandalf denotes our compatible evaluation path
(\S\ref{sec:setup}).}
\label{tab:microbench}
\scriptsize
\setlength{\tabcolsep}{3.0pt}
\renewcommand{\arraystretch}{0.90}
\begin{tabular}{@{}lrrr@{}}
\toprule
Operation & med. & p05 & p95 \\
\midrule
ML-KEM Gen.   & .0856 & .0499 & .1845 \\
ML-KEM Enc.   & .0495 & .0493 & .0921 \\
ML-KEM Dec.   & .0822 & .0578 & .2234 \\
X25519 Gen.   & .0400 & .0329 & .0648 \\
X25519 DH     & .0633 & .0631 & .1483 \\
Src.Get       & .0009 & .0009 & .0009 \\
$\skPRF$      & .0081 & .0081 & .0146 \\
\Ascon        & .0004 & .0004 & .0004 \\
\midrule
\Gandalf-compatible Gen.
               & 38.0172 & 17.8105 & 107.9008 \\
\Gandalf-compatible Sign
               & 18.6147 & 16.0403 & 69.9487 \\
\Gandalf-compatible Verify
               & .9725 & .9291 & 1.6522 \\
\FalconRS Gen.
               & 50.0985 & 20.8594 & 147.6070 \\
\FalconRS Sign
               & 9.1987 & 7.8058 & 13.5069 \\
\FalconRS Verify
               & 1.0181 & .8572 & 1.8336 \\
\bottomrule
\end{tabular}
\end{table}

\subsection{Results}\label{sec:results}
\paragraph{Computation:}
Table~\ref{tab:microbench} shows that ring signing is the dominant local primitive. Median ML-KEM encapsulation/decapsulation is $0.0495/0.0822$\,ms, X25519 DH is $0.0633$\,ms, and the split-key PRF is $0.0081$\,ms, compared with $9.20/1.02$\,ms sign/verify for \FalconRS and $18.61/0.97$\,ms for the \Gandalf-compatible path. 

\paragraph{Complete server path:}
Across $500$ successful source-backed \FalconRS sessions, median server time is $16.29$\,ms (p05 $11.24$\,ms; p95 $24.06$\,ms). Component medians include $1.56$\,ms ring verification, $0.93$\,ms Finished AEAD, $0.15$\,ms ephemeral generation, $0.13$\,ms ML-KEM decapsulation, and $0.09$\,ms X25519. The total also includes serialization, framing, control flow, and local socket I/O.
\vspace{-4mm}
\paragraph{End-to-end interpretation:}
The integrated run measures the complete accepted server path, including
source lookup, signature verification, split-key derivation, Finished
processing, framing, and socket I/O. Component timings in this server
loop are measured in a different execution context from the isolated
microbenchmarks, so we report the two experiments separately rather than
using isolated medians to predict end-to-end latency. Only $\kid_3$ is
serialized, so increasing the source from 32 to 48 bytes does not change
the wire format.

\begin{table}[t]
\centering
\caption{Exact source-mode wire accounting (bytes); ring verification keys are provisioned out of band.}
\label{tab:handshake}
\scriptsize\setlength{\tabcolsep}{2.4pt}\renewcommand{\arraystretch}{0.91}
\begin{tabular}{@{}lrr@{}}
\toprule
Quantity & \Gandalf & \FalconRS \\
\midrule
Hello payload & 1238 & 1238 \\
Client-auth payload & 2395 & 2447 \\
Server / client Finished & 2453 / 2453 & 2453 / 2453 \\
Protocol payload & 8539 & 8591 \\
Frame headers / TCP app. & 16 / \textbf{8555} & 16 / \textbf{8607} \\
Transcript / signature & 2397 / 1236 & 2397 / 1288 \\
\bottomrule
\end{tabular}
\vspace{-2mm}
\end{table}
\vspace{-4mm}
\paragraph{Communication:}
The authenticated transcript is $2397$,B for both instantiations. The
$52$-byte difference in TCP application data is exactly the difference
between the two-member signature formats, so the source-backed exchange
remains below $8.7$,KB in both evaluated configurations. These
measurements characterize the implementation cost of the allocation
realized by \HDKEM rather than claiming a performance advantage over
conventional AKEs: ring-signature authentication dominates local cost,
while the source contributes only its public identifier to the wire
format. Physical QKD latency, WAN behavior, constrained platforms, and
traffic-analysis leakage are outside this experiment.

\section{Discussion}\label{sec:discussion}
\paragraph{Authentication and deniability:}
The allocation theorem separates two questions that are often conflated: who can produce an accepting authenticator and how closely a simulator can match the real signing distribution. \HDKEM keeps the authentication set fixed and pays the witness-hiding gap of the ring signature. A varying-ring design such as \RingXKEM takes the complementary route and can support global denial without comparing two witnesses for the same statement. Neither choice dominates universally; the appropriate allocation depends on the deniability scope required by the application.

\paragraph{Role of the source:}
The third slot changes the exposure trade-off rather than the
deniability mechanism. Deniability remains inherited from the ring
signature, while $k_3$ supplies the only surviving confidentiality
input in C2--C4. Consequently, source-backed \HDKEM obtains C1--C4
coverage without introducing a static KEM/DH confidentiality key;
compromise instead follows the lifetime of the pairwise source.
Source-less \HDKEM gives up C2--C4.

Table~\ref{tab:staticcmp} summarizes the trade-off: static KEM/DH
secrets cover all eCK cases without an out-of-band source, whereas
\HDKEM assigns that role to a pairwise source whose exposure follows
source rotation.

\begin{table}[t]
\centering
\caption{Static--ephemeral trade-offs. Exposure follows the lifetime of the surviving non-ephemeral secret.}
\label{tab:staticcmp}
\scriptsize\setlength{\tabcolsep}{2.1pt}\renewcommand{\arraystretch}{0.92}
\begin{tabular}{@{}p{1.55cm}p{.55cm}p{.76cm}p{1.0cm}p{1.33cm}@{}}
\toprule
Design & eCK & OOB secret & Retroactive secret & Exposure \\
\midrule
\HDKEM, \textsc{src} & C1--4 & pairwise $k_3$ & $k_3$ & source epoch \\
\HDKEM, no source & C1 & none & --- & none \\
X3DH/PQXDH & C1--4 & none & static DH/KEM & identity-key rotation \\
KEMTLS & C1--4 & none & static KEM & cert. rotation \\
NAXOS-style & C1--4 & none & long-term key & key rotation \\
\bottomrule
\end{tabular}
\end{table}

\section{Conclusion}\label{sec:conclusion}

We formulate fresh single-ring deniable authentication as an authentication-allocation problem: simulation may rely on a shared witness, a distinct witness for the same ring, or a witness for a different verified ring. For fixed-ring protocols, our allocation theorem connects the witness-indistinguishability gap $\epsWI$ to the achievable $(\mu,\delta)$ deniability parameters.

\HDKEM instantiates the fixed-ring witness-hiding design with an ML-KEM-768/X25519/source combiner and achieves eCK key indistinguishability, forward secrecy, and strong local deniability under the stated component assumptions. Our evaluation shows that ring-signature authentication dominates local cryptographic cost, while the source-backed handshake requires less than $8.7$,KB of TCP application data. These results characterize the cost of preserving endpoint-only authentication while retaining efficient local deniability in a post-quantum hybrid key exchange.

\clearpage
\appendix

%-------------------------------------------------------------------------------
% Required / recommended submission appendices
%-------------------------------------------------------------------------------

\section*{Ethical Considerations}\label{app:ethics}
This work studies cryptographic protocol design and evaluates locally executed
implementations. It does not involve human subjects, personally identifiable
information, user studies, or measurements of third-party systems. The
experiments use locally generated cryptographic inputs and synthetic protocol
sessions. The mechanisms are intended for defensive key establishment; the
paper does not introduce an operational attack technique or collect data from
deployed users or networks. 

\section*{Open Science}\label{app:openscience}
An anonymized repository accompanies the submission and contains the \HDKEM
source code, build instructions, benchmark programs, CryptoVerif and ProVerif models, verifier outputs, and scripts used for the measurements in
\S\ref{sec:eval}:
\url{https://anonymous.4open.science/r/HDKEM-SEC27/}.
The repository records the commands needed to reproduce the formal checks, primitive benchmarks, server experiment, and exact wire-size accounting.
The artifact contains no human-subject data. The anonymous link is intended
to remain available throughout the review period.

%-------------------------------------------------------------------------------
% References
%-------------------------------------------------------------------------------

\clearpage

%-------------------------------------------------------------------------------
% Optional technical appendices
%-------------------------------------------------------------------------------

\section{Design Rationale and Static--Ephemeral Comparison}\label{app:designmoved}
\subsection{Rationale for a Third Slot and Static--Ephemeral Trade-off}\label{sec:rationale}

\paragraph{Why a third slot:}\label{sec:whythird}

A two-slot combiner using a classical NIKE and a post-quantum KEM already preserves quantum
confidentiality through the KEM slot, so $k_3$ is not needed for post-quantum confidentiality alone.
Its value is instead twofold.

First, with a fully ephemeral KEM/NIKE core, $k_3$ is the only surviving confidentiality slot in
three of the four eCK cases (\S\ref{sec:eckcases}); \S\ref{sec:staticcmp} compares this choice with
designs retaining long-term KEM or DH secrets.

Second, $k_3$ reduces reliance on the concrete ML-KEM/QROM bound. At $\qh=2^{60}$, ML-KEM's
One-Way-to-Hiding loss is approximately $2^{-98}$. The benefit, however, is limited by the split-key
PRF analysis: under the multi-oracle bound, a $256$-bit source gives
$\eps_{\mathsf{prf}}\approx2^{-45.5}$, which dominates the combiner, while $\beta=384$ restores
headroom below the ML-KEM term. This is inexpensive for a sufficiently long pre-shared key, while for
QKD the additional concrete margin depends on the realization error
$\eps_{\mathsf{src}}=\eps_{\mathsf{qkd}}$ of the deployed QKD service.

The source slot is deliberately independent of deniability. Deniability is provided entirely by the
ring signature; making it depend on $k_3$ would remove deniability in source-less mode and duplicate
the role of the public-key authenticator.

\paragraph{Static--ephemeral comparison:}\label{sec:staticcmp}

Protocols that retain a long-term KEM or DH secret, including X3DH/PQXDH, KEMTLS, and NAXOS-style
designs, can cover all four eCK cases without an out-of-band source. This is a genuine advantage over
\HDKEM's source-less mode. The corresponding cost is a longer-lived cryptographic secret whose later
compromise may expose a larger retrospective window.

Table~\ref{tab:staticcmp} summarizes the main trade-off; Appendix~\ref{app:staticcmp} expands the exposure-window calculation for representative key-rotation schedules. Because our theorems are stated over
abstract KEM, NIKE, signature, and source interfaces, compatible primitive replacements require no
change to the proof structure.

\section{Additional Evaluation Context}\label{app:evalextras}

\subsection{Case Study: Legacy Infrastructure Protocols}\label{sec:casestudy}

As a deployment case study, we consider protecting ANSI~C12.22, DNP3, IEC~61850, and
IEEE~C37.118.2 payloads with Ascon-128a under $\ksh$. Their $30$--$51$-byte messages represent a
setting in which \HDKEM's provisioning assumptions are particularly plausible: infrastructure
devices may remain deployed for ten to thirty years, making long-horizon confidentiality and
deniability relevant, while pairwise secrets can be installed at commissioning.

On our reference platform, an \Ascon-128a operation takes less than a microsecond, so payload
protection is negligible relative to ring-signature authentication in Table~\ref{tab:microbench}. We leave
protocol-specific end-to-end latency measurements to future evaluation.

\begin{table}[!t]
\centering
\caption{Authentication allocation of related constructions.}
\label{tab:baseline}
\scriptsize\setlength{\tabcolsep}{2.5pt}\renewcommand{\arraystretch}{0.92}
\begin{tabular}{@{}lllc@{}}
\toprule
Scheme & Deniability basis & Alloc. & Static \\
\midrule
\XWing & none & --- & no \\
PQ-TLS + ML-DSA-44 & none & --- & no \\
KEMTLS & none & --- & yes \\
\Shadowfax & ring anonymity & hiding & no \\
\RingXKEM (\Gandalf) & anonymity, $2^{-30}$ & varying ring & no \\
\RingXKEM (\FalconRS) & $(1{+}2^{-27},2^{-57})$ & varying ring & no \\
\bottomrule
\end{tabular}
\end{table}

\begin{table}[!t]
\centering
\caption{Representative constructions by authentication allocation.}
\label{tab:cells}
\scriptsize\setlength{\tabcolsep}{2.4pt}\renewcommand{\arraystretch}{0.92}
\begin{tabular}{@{}p{1.05cm}p{1.70cm}p{3.55cm}@{}}
\toprule
Allocation & Construction & Basis \\
\midrule
Witness-sharing & MAC-auth. AKE & peers share the authentication witness \\
Witness-hiding & \HDKEM & distinct endpoint ring witnesses; Thm.~\ref{thm:deny} \\
Varying ring & \RingXKEM~\cite{Hashimoto25} & auxiliary unauthenticated member changes the verified ring \\
\bottomrule
\end{tabular}
\end{table}

\section{Constructions by Allocation}\label{app:cells}
Table~\ref{tab:cells} gives representative placements that follow directly from the authentication mechanism: shared-MAC AKE is witness-sharing, \HDKEM is witness-hiding by Theorem~\ref{thm:deny}, and the auxiliary-member form of \RingXKEM is varying-ring. Table~\ref{tab:baseline} compares additional constructions without claiming an exhaustive formal classification.

\section{Static-Ephemeral Comparison and Proof Portability}\label{app:staticcmp}
This appendix expands \S\ref{sec:staticcmp}. In a static-ephemeral design the static secret is the
surviving slot in exactly the cases where the ephemeral material is compromised (C2--C4), so its
later compromise retroactively removes the guarantee for every past session back to the last
rotation. \HDKEM's third slot plays the same role in the same cases, but $k_3$ is re-derived per
epoch (\S\ref{sec:threat}), so a single compromise reaches only one epoch's sessions. Writing
$T_{\mathrm{lt}}$ for the static key's lifetime and $T_{\mathrm{ep}}$ for the epoch length, a
static-ephemeral design's exposure window exceeds \HDKEM's by $T_{\mathrm{lt}}/T_{\mathrm{ep}}$
sessions: for the deployments of \S\ref{sec:casestudy}, where certificates run ten to thirty years
against a monthly-rotated device secret, that factor is $10^2$--$10^3$; for consumer messaging with
month-scale identity-key rotation it is close to $1$ and the static-ephemeral design is preferable,
since it needs no provisioning. \HDKEM's requirement is also strictly harder to meet
operationally---a pairwise $32$--$48$ byte secret must exist at commissioning---which holds for
individually commissioned infrastructure endpoints and generally does not for open-enrollment
deployments; \S\ref{sec:staticcmp} favors the static-ephemeral design over source-less \HDKEM in
that case, since source-less mode gives up C2--C4 entirely.

\paragraph{Proof portability:} Every theorem is stated against abstract notions---IND-CCA2, Gap-DH,
EUF-CMA, $(\mu,\delta)$-deniability, split-key PRF, composable source, AEAD---so any qualifying
component substitutes with proofs intact. This is also what makes Remark~\ref{rem:twobounds} a
choice of instantiation rather than a change of theorem: both analyses bound the same abstract
quantity $\eps_{\mathsf{prf}}$, and no proof in \S\ref{sec:sec} distinguishes them. For the published \Gandalf construction, the QROM unforgeability ingredient is supplied by~\cite{Beckmann26}. Remark~\ref{rem:quantumdeny} explains why their QROM anonymity result does not by itself establish the stronger deniability ingredient of Theorem~\ref{thm:deny}. The performance configuration described in \S\ref{sec:eval} is separate from the published \Gandalf security reduction.

\section{Protocol Message Flow}\label{app:flow}
Figure~\ref{fig:flow} summarizes the \textsc{src}-mode handshake referenced from \S\ref{sec:protospec}. Ring keypairs are long-term and provisioned once; only the client signs in this mode.

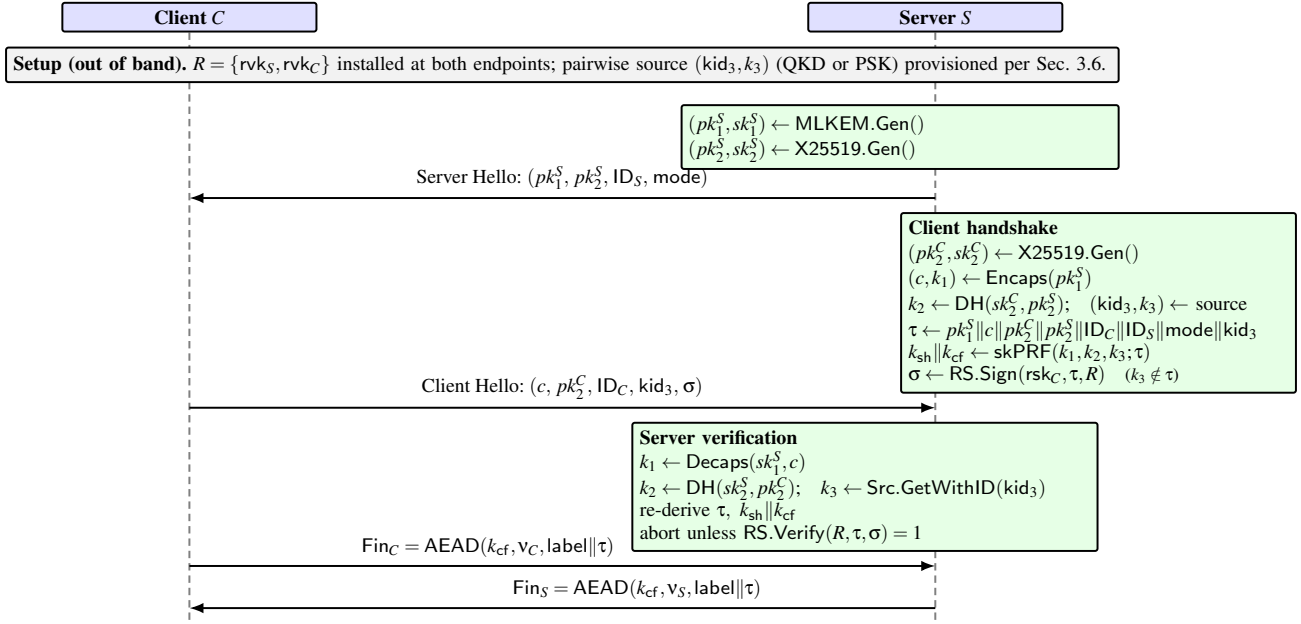
\begin{figure*}[t]
\centering
\begin{adjustbox}{max width=.96\textwidth,center}
\begin{tikzpicture}[
  font=\fontsize{8}{9}\selectfont,
  hdr/.style={draw,thick,rounded corners=1pt,fill=blue!12,
              text width=3.4cm,align=center,inner sep=3pt,
              font=\bfseries\fontsize{8}{9}\selectfont},
  stepbox/.style={draw,thick,rounded corners=1pt,align=left,
                  text width=6.0cm,inner sep=3pt},
  cbox/.style={stepbox,fill=green!10,anchor=north west},
  sbox/.style={stepbox,fill=green!10,anchor=north east},
  vbox/.style={stepbox,fill=black!5},
  msg/.style={-{Latex[length=2mm]},thick},
  lifeline/.style={densely dashed,thick,black!50},
  srcline/.style={densely dotted,very thick,blue!70!black},
]
\pgfdeclarelayer{bg}\pgfsetlayers{bg,main}
\coordinate (CX) at (0,0);
\coordinate (SX) at (10.6,0);

\node[hdr] (Chead) at (CX) {Client $C$};
\node[hdr] (Shead) at (SX) {Server $S$};

\node[vbox, text width=15.6cm, below=2mm of $(Chead.south)!0.5!(Shead.south)$,
      anchor=north] (setup)
  {\textbf{Setup (out of band).} $R=\{\rvk_S,\rvk_C\}$ installed at both endpoints;
   pairwise source $(\kid_3,k_3)$ (QKD or PSK) provisioned per Sec.~3.6.};

\node[sbox, below=3mm of setup.south east, anchor=north east, text width=6.0cm] (Sgen)
  {$(pk_1^S,sk_1^S)\gets\mathsf{MLKEM.Gen}()$\\
   $(pk_2^S,sk_2^S)\gets\mathsf{X25519.Gen}()$};

\coordinate (helloY) at ($(Sgen.south)+(0,-4mm)$);
\draw[msg] (SX|-helloY) -- node[above,pos=0.5,align=center]
  {\footnotesize Server Hello: $(pk_1^S,\,pk_2^S,\,\ID_S,\,\mathsf{mode})$}
  (CX|-helloY);

\node[cbox, below=2mm of helloY, anchor=north west, text width=5.4cm] (Chs)
  {\textbf{Client handshake}\\[1pt]
   $(pk_2^C,sk_2^C)\gets\mathsf{X25519.Gen}()$\\
   $(c,k_1)\gets\mathsf{Encaps}(pk_1^S)$\\
   $k_2\gets\mathsf{DH}(sk_2^C,pk_2^S)$;\quad $(\kid_3,k_3)\gets$ source\\
   $\tau \gets pk_1^S\|c\|pk_2^C\|pk_2^S\|\ID_C\|\ID_S\|\mathsf{mode}\|\kid_3$\\
   $\ksh\|\kconf\gets\skPRF(k_1,k_2,k_3;\tau)$\\
   $\sigma\gets\mathsf{RS.Sign}(\rsk_C,\tau,R)$\quad {\scriptsize($k_3\notin\tau$)}};

\coordinate (cliHelloY) at ($(Chs.south)+(0,-2mm)$);
\draw[msg] (CX|-cliHelloY) -- node[above,pos=0.5,align=center]
  {\footnotesize Client Hello: $(c,\,pk_2^C,\,\ID_C,\,\kid_3,\,\sigma)$}
  (SX|-cliHelloY);

\node[sbox, below=2mm of cliHelloY, anchor=north east, text width=6.4cm] (Sver)
  {\textbf{Server verification}\\[1pt]
   $k_1\gets\mathsf{Decaps}(sk_1^S,c)$\\
   $k_2\gets\mathsf{DH}(sk_2^S,pk_2^C)$;\quad $k_3\gets\mathsf{Src.GetWithID}(\kid_3)$\\
   re-derive $\tau,\ \ksh\|\kconf$\\
   abort unless $\mathsf{RS.Verify}(R,\tau,\sigma)=1$};

\coordinate (finY) at ($(Sver.south)+(0,-2mm)$);
\draw[msg] (CX|-finY) -- node[above,pos=0.4,align=center]
  {\footnotesize $\mathsf{Fin}_C=\mathsf{AEAD}(\kconf,\nu_C,\mathsf{label}\|\tau)$}
  (SX|-finY);
\coordinate (finY2) at ($(finY)+(0,-6mm)$);
\draw[msg] (SX|-finY2) -- node[above,pos=0.4,align=center]
  {\footnotesize $\mathsf{Fin}_S=\mathsf{AEAD}(\kconf,\nu_S,\mathsf{label}\|\tau)$}
  (CX|-finY2);

\coordinate (bottomY) at ($(finY2)+(0,-2mm)$);
\begin{pgfonlayer}{bg}
\draw[lifeline] (Chead.south) -- (CX|-bottomY);
\draw[lifeline] (Shead.south) -- (SX|-bottomY);
\end{pgfonlayer}
\end{tikzpicture}
\end{adjustbox}
\caption{\HDKEM message flow, \textsc{src} mode ($\vec{k}=(k_1,k_2,k_3)$). $\tau$ binds $c$, both
X25519 values, both identities and $\mathsf{mode}$, but \emph{excludes} $k_3$: ring signatures hide
the signer's identity, not the message content, so $\sigma$ would expose a source secret placed
inside $\tau$. $k_3$ enters only the split-key PRF deriving $\ksh,\kconf$ in Algorithm~\ref{alg:handshake},
keeping the eCK C2--C4 argument and Thm.~\ref{thm:deny} sound.}
\label{fig:flow}
\end{figure*}

\section{Deniability Game (Algorithm)}\label{app:dengame}
This appendix gives the full pseudocode for the game referenced by Definition~\ref{def:mudelta}.

\begin{algorithm}[t]
\caption{Deniability game. Steps marked $\star$ apply only to the global variant.}
\label{alg:dengame}
\begin{algorithmic}[1]
\Function{Game}{$1^\lambda, \mathsf{mode}$}
  \State $\mathcal{C} := \mathcal{N}\setminus\mathcal{H}$
  \For{$u \in \mathcal{N}$} $(\rvk_u,\rsk_u)\gets\mathsf{RS.KeyGen}(1^\lambda)$ \EndFor
  \For{each pair $(u,v)$} provision $k_3^{(u,v)}$ \EndFor
  \State $\mathsf{st}_\mathcal{D}\gets\mathcal{D}^{\mathcal{O}}((\rvk_u)_{u\in\mathcal{N}})$
  \State $\mathsf{lk}\gets\Leak((\rsk_u,k_3^{(u,\cdot)},\mathsf{st}_u)_{u\in\mathcal{H}})$
  \State $\mathsf{dc}\gets\Disc((\rsk_v,k_3^{(v,\cdot)},\mathsf{st}_v)_{v\in\mathcal{C}})$
  \State \Return $b\gets\mathcal{D}(\mathsf{st}_\mathcal{D},\mathsf{lk},\mathsf{dc})$
\EndFunction
\Statex
\Function{$\mathcal{O}$}{$s,r$}
  \State \textbf{require} $s\neq r$ and $(s,r)\notin\mathcal{C}\times\mathcal{C}$
  \State \textbf{require} $(s,r)\notin\mathcal{H}\times\mathcal{H}$ unless global$^\star$
  \If{$\mathsf{mode}=\real$}
    \State $(K,\rho)\gets\Pi.\mathsf{Run}(s,r)$
  \Else
    \If{$(s,r)\in\mathcal{H}\times\mathcal{C}$}
      \State $(K,\rho)\gets\Simm(\rsk_r,k_3^{(s,r)},(s,r))$
    \ElsIf{$(s,r)\in\mathcal{C}\times\mathcal{H}$}
      \State $(K,\rho)\gets\Simm(\rsk_s,k_3^{(s,r)},(s,r))$
    \Else \Comment{global only$^\star$}
      \State $(K,\rho)\gets\Simm(\bot,\bot,(s,r))$
    \EndIf
  \EndIf
  \State \Return $(K,\rho)$
\EndFunction
\end{algorithmic}
\end{algorithm}

\section{Full Proofs Deferred from the Main Text}\label{app:fullproofs}
This appendix gives the full proofs sketched in the body, in the order the corresponding statements
appear.

\paragraph{Lemma~\ref{lem:skprf} (Split-key PRF over a min-entropy slot):}
Distinguishing requires the distinguisher's queries to carry non-negligible amplitude on the
challenge input, hence to find the withheld $\kappa_i$. By independence the search is unstructured,
so the generic-search bound applies: for an oracle whose marking probability is at most
$\lambda_{\mathsf{mk}}$ at every point, no $q$-query quantum algorithm succeeds with probability
better than $8(q+1)^2\lambda_{\mathsf{mk}}$~\cite[Thm.~1]{HRS16}, and by hypothesis
$\lambda_{\mathsf{mk}} \le 2^{-\beta}$. Optimality of quantum search on unstructured
domains~\cite{Bennett97} shows the quadratic dependence on $q$ cannot be improved.

\paragraph{Proposition~\ref{prop:containment} (Containment):}
Let $Q \in \Simm_P$, let $\mathsf{acc}$ be the honest acceptance predicate, and let $\mathcal{D}$ be
the efficient distinguisher that outputs $0$ exactly when $\mathsf{acc}$ accepts. Transcripts
genuinely produced by $P$ are accepted with overwhelming probability, so
$\Pr[\mathcal{D}(\real) = 0] \ge 1 - \mathsf{negl}$. Applying Definition~\ref{def:hs} to $\mathcal{D}$
gives $1 - \mathsf{negl} \le \mu\cdot\Pr[\mathcal{D}(\sm) = 0] + \delta$, hence
$\Pr[\mathsf{acc}(\sm)] \ge (1 - \mathsf{negl} - \delta)/\mu$, which is non-negligible for every
polynomial $\mu$ and negligible $\delta$. So $Q$'s outputs are accepted non-negligibly often, and
$Q \in \Auth$.

\paragraph{Theorem~\ref{thm:alloc} (Allocations in $\mathcal{F}$ are exhaustive):}
By Proposition~\ref{prop:containment}, $Q \in \Auth$, so $Q$ produces accepting transcripts with
non-negligible probability $p$. We first show $Q$ holds a ring secret key.
Suppose not. Build a forger $\mathcal{B}$ against Definition~\ref{def:eufcma} for the ring $R$, which
is fixed and known by F2. $\mathcal{B}$ receives the ring's verification keys and a signing oracle. It
runs $Q$'s simulator, answering every request $Q$ makes for an honest party's signature by querying
the oracle and recording the queried messages in a set $\mathcal{Q}$. When the simulator outputs a
transcript containing $(\tau^*, \sigma^*)$, $\mathcal{B}$ outputs it. By F3, $\tau^*$ differs from
every $\tau$ in $\mathcal{Q}$ except with negligible probability: each contains a distinct fresh
component, and a repeated $\tau^*$ would be rejected by an honest verifier, contradicting acceptance.
So $\mathcal{B}$ outputs a valid signature on an unqueried message with probability at least
$p - \mathsf{negl}$, whence $p \le \eps_{\mathsf{sig}} + \mathsf{negl}$. This contradicts
non-negligible $p$ up to $\eps_{\mathsf{sig}}$, so $Q$ holds a ring secret key.
If the key $Q$ holds is one $P$ also holds, case (i) applies: since $Q \neq P$, that key is available
to a principal other than its legitimate holder, and $\Auth$ contains that principal.
Otherwise $Q$ signs under a key distinct from $P$'s. Let $T$ be the fresh,
signer-independent transcript-message distribution induced by $\Pi$, and let
$\mathcal{D}^*$ attain $\epsWI$ for the joint distributions
$(\tau,\sigma)$ of Definition~\ref{def:wi}. The transcript distinguisher
$\mathcal{D}$ performs no simulation and loses no advantage: it extracts
$(\tau,\sigma)$ and runs $\mathcal{D}^*$. Orienting the output if necessary,
write $a^\star=\Pr[\mathcal{D}(\real)=0]$ and
$b^\star=\Pr[\mathcal{D}(\sm)=0]$ with
$\epsWI=a^\star-b^\star$. Definition~\ref{def:hs} gives
$a^\star\le\mu b^\star+\delta$, hence
$\epsWI\le(\mu-1)b^\star+\delta\le(\mu-1)+\delta$.
Thus the witness-to-transcript reduction is lossless; $\eps_{\mathsf{sig}}$
appears only in excluding the preceding no-key branch. The final relaxation is
black-box tight when $b^\star=1$. Any AOS/RPSF-specific improvement must
therefore come from additional information about the signing distributions,
not from the allocation argument itself.

\paragraph{Theorem~\ref{thm:nike} (NIKE-slot soundness):}
Substituting either X25519 public value changes $\tau$ and hence requires a forgery for $R$, at cost
$\eps_{\mathsf{sig}}$; conditioned on no forgery, both parties hold honestly generated values. For
pseudorandomness, a reduction $\mathcal{B}$ receiving a Gap-DH challenge $(g^a,g^b)$ sets the two
X25519 public values to $g^a$ and $g^b$, generates \emph{both} ring keypairs itself, and answers all
signing queries honestly for either member. In \HDKEM only the client signs (\S\ref{sec:protospec}),
so concretely $\mathcal{B}$ signs the test transcript itself under $\rsk_C$, which it generated; the
point of holding both ring secrets, rather than a proof obligation of the protocol, is that
$\mathcal{B}$ never has to decide in advance which ring member's key the eventual forgery-avoidance
argument of \S\ref{sec:eckcases} will need, so the same reduction serves Theorem~\ref{thm:eck}
unchanged. $\mathcal{B}$ never queries its own EUF-CMA oracle on the challenge message, so holding
both ring secrets introduces no circularity. At the test session $k_2=g^{ab}$ is $\mathcal{B}$'s
output. The oracle-sequencing detail is in Appendix~\ref{app:nike}.

\paragraph{Theorem~\ref{thm:eck} (eCK security):}
Games $G_0\to G_4$. \emph{$G_1$ (reject forgeries).} Abort if an honest session accepts a signature not
produced by an honest party for $R$. Since $\tau$ binds both X25519 values, $pk_1^S$, $c$, both
identities and the mode, any active modification requires a forgery.
$|\Pr[W_1]-\Pr[W_0]|\le\eps_{\mathsf{sig}}$.
\emph{$G_2$ (reject confirmation forgeries and desynchronization).} Abort if an honest session accepts
a Finished message under a $\kconf$ not honestly derived by the intended partner. $\kconf$ derives from
$(k_1,k_2,k_3;\tau)$ with $\tau$ containing $c$, unique per session, so two honest sessions deriving
the same $\kconf$ from different transcripts requires a $\skPRF$ collision or a source-key collision.
$\tau$ does not itself contain $k_3$ (\S\ref{sec:protospec}), so this hop does not rely on $\sigma$
hiding $k_3$; it relies only on $\tau$'s freshness, a public property of $c$.
$|\Pr[W_2]-\Pr[W_1]|\le\eps_{\mathsf{prf}}+\eps_{\mathsf{aead}}$.
\emph{$G_3$ (reduce to a surviving slot).}
Map the test session to Table~\ref{tab:eckcases}. In C1,
\[
  \eps_{\mathsf{slot}}
  =\min\{\eps_{\mathsf{mlkem}},\eps_{\mathsf{dh}},\eps_{\mathsf{src}}\},
\]
while C2--C4 leave only $k_3$ and hence
$\eps_{\mathsf{slot}}=\eps_{\mathsf{src}}$. Therefore
$|\Pr[W_3]-\Pr[W_2]|\le\eps_{\mathsf{slot}}$.

\emph{$G_4$ (idealize the derived keys).}
Apply split-key PRF security to the surviving slot at cost
$\eps_{\mathsf{prf}}$; a single call derives both $\ksh$ and $\kconf$.
Thus $\Pr[W_4]=1/2$.

\paragraph{Theorem~\ref{thm:fs} (Forward secrecy):}
After ephemeral erasure, later disclosure of the ring-signing keys reveals none of
$k_1,k_2,k_3$; the robust-combiner reduction gives the bound of
Theorem~\ref{thm:fs}.

\paragraph{Theorem~\ref{thm:deny} (Strong local deniability):}
Fix an accused $P$ and an accuser $Q$ holding $\rsk_Q$ and $k_3^{(P,Q)}$. In \textsc{src} mode only the
client signs, so take $P = C$; the case $P = S$ is immediate, since the transcript's only signature is
then produced by $Q$ itself and every other component is a deterministic function of $\rsk_Q$,
$k_3^{(P,Q)}$ and public data.
Define $\Simm$: sample a fresh X25519 keypair to play the client's; encapsulate to the server's ML-KEM
public key, which is public, obtaining $(c,k_1)$; compute $k_2$ from the fresh key and the server's
X25519 value; form $\tau$ exactly as in Algorithm~\ref{alg:handshake} --- note $\tau$ needs only public values and does not
require $k_3^{(P,Q)}$ to construct, since $k_3\notin\tau$; derive $\ksh\|\kconf$ by the same split-key PRF step
using $k_3^{(P,Q)}$; sign $\sigma\gets\mathsf{RS.Sign}(\rsk_Q,\tau,R)$ under its own ring key; and
produce the Finished messages under $\kconf$. $\Simm$ uses only $\rsk_Q$, $k_3^{(P,Q)}$ and public
data, so it is runnable by $Q$ as Definition~\ref{def:mudelta} requires.
Three hops separate the real game from the simulated one. \emph{$H_0$} is the real game. \emph{$H_1$}
replaces the client's ephemeral X25519 keypair and encapsulation by ones sampled by $\Simm$. Both are
public-coin operations with no secret input, so the distributions are identical and
$\Pr[H_1]=\Pr[H_0]$. \emph{$H_2$} replaces $\tau$, $\ksh$ and $\kconf$ by $\Simm$'s. Each is a
deterministic function of the values fixed in $H_1$, of public identifiers, and of $k_3^{(P,Q)}$, all
of which $\Simm$ holds, so $\Pr[H_2]=\Pr[H_1]$. \emph{$H_3=\sm$} replaces
$\sigma\gets\mathsf{RS.Sign}(\rsk_P,\tau,R)$ by $\sigma\gets\mathsf{RS.Sign}(\rsk_Q,\tau,R)$. This is
the only remaining difference, and it is exactly what Definition~\ref{def:rsden} quantifies: same
ring, same message, different signing key. Because $\tau$ excludes $k_3$, this is a well-formed
instance of Definition~\ref{def:rsden}'s game regardless of what either party's $k_3$ is; had $\tau$
contained $k_3^{(P,Q)}$, the two signing distributions being compared would additionally depend on a
secret the deniability game's signing oracle does not take as input, and the reduction below would not
go through. A distinguisher $\mathcal{D}$ separating $H_2$ from $H_3$ by more than the
$(\mu_{\mathsf{RS}},\delta_{\mathsf{RS}})$ relation yields an adversary $\mathcal{A}_{\mathsf{RS}}$
against Definition~\ref{def:rsden}: $\mathcal{A}_{\mathsf{RS}}$ receives all keypairs, generates every
other transcript component itself as in $H_2$, obtains $\sigma$ from the signing oracle on
$(\tau, u_0=P, u_1=Q)$, hands the assembled transcript to $\mathcal{D}$, and echoes its output. Hence
$\Pr[H_2=0]\le\mu_{\mathsf{RS}}\cdot\Pr[H_3=0]+\delta_{\mathsf{RS}}$. Chaining the three hops gives the
claimed relation between $\real$ and $\sm$.
At $\mathsf{leak}=\mathsf{high}$, $\mathcal{D}$ receives $\rsk_P$, $k_3^{(P,Q)}$ and all of $P$'s
state. The source secret and the state are used identically in $H_2$ and $H_3$, so they cannot
separate them. Possessing $\rsk_P$ likewise does not help beyond the ring signature's own bound, because Definition~\ref{def:rsden}'s game already hands the adversary \emph{every} keypair; the reduction above is therefore valid verbatim under maximal leakage, with no additional term. The same
argument covers $\mathsf{disc}=\mathsf{high}$. Finally, $\Simm$ does not depend on the accuser's
behaviour: it uses only $\rsk_Q$, $k_3$ and public data, so a malicious accuser may deviate
arbitrarily and $\Simm$ still runs.

\paragraph{Corollary~\ref{cor:multi} (Multi-session deniability):}
For the direct reduction, a single adversary in Definition~\ref{def:rsden}
uses $Q=\ell$ signing queries, one for each HDKEM session, and generates all
non-signature fields exactly as in Theorem~\ref{thm:deny}. Hence the joint
view inherits $(\mu_\ell,\delta_\ell)$ without an additional hybrid loss.
For the generic bound, define $J_0,\ldots,J_\ell$ where $J_i$ simulates the
first $i$ sessions. Adjacent hybrids differ in one signing witness, so
$\Pr[J_{i-1}=0]\le\mu_1\Pr[J_i=0]+\delta_1$; unrolling gives
$\mu_1^\ell$ and $\delta_1\sum_{j=0}^{\ell-1}\mu_1^j$.
No independence assumption is needed for public cross-session metadata that
is identically distributed in adjacent hybrids and admissible as auxiliary
input to the ring-signature game. Reused signing randomness or
witness-dependent metadata can violate that premise and require an
additional primitive-level analysis. We prove no matching lower bound for
the geometric expression.

\paragraph{Proposition~\ref{prop:noglobal} (No global deniability):}
Global deniability requires a simulator holding no party's secrets to produce an accepting transcript.
An accepting transcript contains a signature valid for $R=\{\rvk_S,\rvk_C\}$ on a $\tau$ fresh by F3;
producing one without $\rsk_S$ or $\rsk_C$ contradicts Definition~\ref{def:eufcma} by the forger
construction in the proof of Theorem~\ref{thm:alloc}. Hence simulated acceptance is negligible while
real acceptance is overwhelming, which no polynomial $\mu$ and negligible $\delta$ reconcile.

\section{Full Proof of the Robust Combiner}\label{app:combproof}

We prove Theorem~\ref{thm:comb}. The proof is given for a single test session first and
then extended to multiple sessions.

\paragraph{Experiment:}
Let the test session derive
\[
  K = F(k_1,k_2,k_3;\tau),
\]
where $F=\skPRF$, $\tau$ is the injectively encoded public transcript of
Algorithm~\ref{alg:handshake}, and $K$ denotes the concatenation
$\ksh\|\kconf$. The adversary receives the complete public transcript and all
information permitted by the corresponding freshness experiment. It is challenged on
either the real $K$ or a uniformly random string of the same length.

For each $i\in\{1,2,3\}$, let $\mathsf{Sec}_i$ denote the security experiment for the
$i$th slot:
\[
\begin{aligned}
\mathsf{Sec}_1
  &= \mathsf{IND\mbox{-}CCA2}_{\mathsf{MLKEM}},\\
\mathsf{Sec}_2
  &= \mathsf{GapDH},\\
\mathsf{Sec}_3
  &= \mathsf{Source}.
\end{aligned}
\]
Write the corresponding costs as
$\eps_{\mathsf{mlkem}}$, $\eps_{\mathsf{dh}}$, and
$\eps_{\mathsf{src}}$.

We show that, for every slot $i$ that remains secure,
\[
  \adv^{\mathsf{comb}}_{\HDKEM,i}(\mathcal{A})
  \le \eps_i+\eps_{\mathsf{prf}}.
\]

\paragraph{Context binding:}
A point needed in all three cases is that the ML-KEM ciphertext belongs to the
split-key PRF context:
\[
  \tau =
  pk_1^S\|c\|pk_2^C\|pk_2^S\|\ID_C\|\ID_S\|\mathsf{mode}\|\kid_3.
\]
Hence, for any $c'\neq c$,
\[
  \tau(c')\neq\tau(c)
\]
by injectivity of the transcript encoding; $\kid_3$ is an additional public
context field and does not weaken this implication. This remains true even in the exceptional
case
\[
  \mathsf{Decaps}(sk_1^S,c')
  =
  \mathsf{Decaps}(sk_1^S,c).
\]
Therefore a second ciphertext yielding the same ML-KEM shared secret does not reproduce
the challenge input to $F$: it produces
\[
  F(k_1,k_2,k_3;\tau(c'))
\]
rather than
\[
  F(k_1,k_2,k_3;\tau(c)).
\]
This is why HDKEM does not require ciphertext second-preimage resistance. The
second-preimage attack on a KEM combiner applies when the modified ciphertext is omitted
from the combiner context; optimized constructions such as \XWing therefore establish
C2PRI separately~\cite{Connolly24}.

\paragraph{Case 1: the ML-KEM slot survives:}
Assume $k_1$ is the surviving secure slot. We define hybrids $G_0,G_1,G_2$.

\emph{$G_0$.}
This is the real combiner experiment. The test session contains
\[
  (c^*,k_1^*)\gets
  \mathsf{MLKEM.Encaps}(pk_1^S),
\]
and
\[
  K^*
  =
  F(k_1^*,k_2^*,k_3^*;\tau^*).
\]

\emph{$G_1$.}
Replace $k_1^*$ by an independent uniform value
$u_1$ of the ML-KEM shared-key length while leaving $c^*$ and every other value
unchanged:
\[
  K^*
  =
  F(u_1,k_2^*,k_3^*;\tau^*).
\]

Suppose $\mathcal{A}$ distinguishes $G_0$ from $G_1$ with advantage $\alpha$.
Construct an ML-KEM IND-CCA2 adversary $\mathcal{B}_1$ as follows.

$\mathcal{B}_1$ receives an ML-KEM challenge
\[
  (pk^*,c^*,K_b),
\]
where $K_0$ is the real encapsulated key and $K_1$ is uniform.
It sets the test session's ML-KEM public key to $pk^*$ and ciphertext to $c^*$,
generates the X25519 and source slots itself, forms
\[
  \tau^*
  =
  pk^*\|c^*\|pk_2^C\|pk_2^S\|\ID_C\|\ID_S\|\mathsf{mode},
\]
and returns
\[
  F(K_b,k_2^*,k_3^*;\tau^*)
\]
as the test-session key.

For every adversarial ML-KEM ciphertext $c\neq c^*$,
$\mathcal{B}_1$ forwards $c$ to its IND-CCA2 decapsulation oracle and uses the
returned value when simulating the corresponding protocol session. A query with
$c=c^*$ is disallowed by the IND-CCA2 experiment and corresponds to the test
ciphertext itself.

If $b=0$, $\mathcal{B}_1$ simulates $G_0$; if $b=1$, it simulates $G_1$.
Thus
\[
  \left|
    \Pr[G_0\Rightarrow1]
    -
    \Pr[G_1\Rightarrow1]
  \right|
  \le
  \eps_{\mathsf{mlkem}}.
\]

Notice that no extra assumption is needed for a query $c'\neq c^*$ that happens to
decapsulate to $k_1^*$. Such a query is legal in the IND-CCA2 game; moreover, in
HDKEM it induces $\tau(c')\neq\tau(c^*)$, so it cannot reproduce the test input to
$F$.

\emph{$G_2$.}
Replace
\[
  F(u_1,k_2^*,k_3^*;\tau^*)
\]
by a uniformly random string of equal length.

In $G_1$, the first slot $u_1$ is independent of every other PRF input and hidden
from $\mathcal{A}$. By Definition~\ref{def:skprf},
\[
  \left|
    \Pr[G_1\Rightarrow1]
    -
    \Pr[G_2\Rightarrow1]
  \right|
  \le
  \eps_{\mathsf{prf}}.
\]
Since $G_2$ is the ideal combiner experiment,
\[
  \adv^{\mathsf{comb}}_{\HDKEM,1}(\mathcal{A})
  \le
  \eps_{\mathsf{mlkem}}+\eps_{\mathsf{prf}}.
\]

\paragraph{Case 2: the X25519 slot survives:}
Assume $k_2$ is the surviving secure slot. Define $G_0$ as above and let $G_1$
replace
\[
  k_2^*=\mathsf{DH}(sk_2^C,pk_2^S)
\]
by an independent uniform group-secret value $u_2$.

A distinguisher between these games yields a Gap-DH adversary $\mathcal{B}_2$:
given the corresponding challenge public values, $\mathcal{B}_2$ installs them as
$pk_2^C$ and $pk_2^S$, generates the ML-KEM and source slots itself, and uses the
challenge DH value as the second PRF input. Hence
\[
  \left|
    \Pr[G_0\Rightarrow1]
    -
    \Pr[G_1\Rightarrow1]
  \right|
  \le
  \eps_{\mathsf{dh}}.
\]

Conditioned on $G_1$, $u_2$ is independent of the remaining inputs and of $\tau^*$.
The split-key PRF definition therefore gives
\[
  \left|
    \Pr[G_1\Rightarrow1]
    -
    \Pr[G_2\Rightarrow1]
  \right|
  \le
  \eps_{\mathsf{prf}},
\]
where $G_2$ replaces the derived key by uniform. Thus
\[
  \adv^{\mathsf{comb}}_{\HDKEM,2}(\mathcal{A})
  \le
  \eps_{\mathsf{dh}}+\eps_{\mathsf{prf}}.
\]

Again, arbitrary behavior of the ML-KEM slot does not invalidate the reduction.
In particular, if an attacker changes $c$ while preserving its decapsulated value,
the resulting transcript context changes from $\tau(c)$ to $\tau(c')$ and therefore
constitutes a different invocation of $F$.

\paragraph{Case 3: the source slot survives:}
For QKD, replace $k_3^*$ by an ideal uniform source at cost
$\eps_{\mathsf{ks}}$ and then apply split-key PRF security. For a
pre-shared source satisfying the stated independence and min-entropy
premise, Lemma~\ref{lem:skprf} applies directly, so there is no separate
source-idealization loss. Hence in either case

$$
  \adv^{\mathsf{comb}}_{\HDKEM,3}(\mathcal{A})
  \le
  \eps_{\mathsf{src}}+\eps_{\mathsf{prf}}.
$$

Changing $c$ changes $\tau$, so even a compromised ML-KEM slot cannot
reproduce the test PRF input with a different ciphertext.

\paragraph{Combining the cases:}
Let $I_{\mathsf{sec}}\subseteq\{1,2,3\}$ be the set of slots whose corresponding
security assumption holds. Since the above reduction may be instantiated with any
$i\in I_{\mathsf{sec}}$,
\[
  \adv^{\mathsf{comb}}_{\HDKEM}(\mathcal{A})
  \le
  \min_{i\in I_{\mathsf{sec}}}
  \bigl(\eps_i+\eps_{\mathsf{prf}}\bigr)
  =
  \min_{i\in I_{\mathsf{sec}}}\eps_i
  +\eps_{\mathsf{prf}}.
\]

\paragraph{Multiple sessions:}
For $\ell$ sessions, define hybrids
$H_0,\ldots,H_\ell$, where $H_j$ idealizes the selected surviving slot in the
first $j$ sessions and leaves sessions $j+1,\ldots,\ell$ real. Adjacent hybrids
differ in exactly one session. Applying the corresponding single-session reduction
to that session gives
\[
  \left|
    \Pr[H_{j-1}\Rightarrow1]
    -
    \Pr[H_j\Rightarrow1]
  \right|
  \le
  \eps_i+\eps_{\mathsf{prf}}.
\]
The triangle inequality yields
\[
  \adv^{\mathsf{comb},\ell}_{\HDKEM,i}(\mathcal{A})
  \le
  \ell\bigl(\eps_i+\eps_{\mathsf{prf}}\bigr).
\]
The reduction changes one fixed session per hybrid and therefore requires no
guessing step inside an individual hop; the overall multi-session loss is the
standard additive hybrid loss.
\qed

\section{Oracle Sequencing for Theorem~\ref{thm:nike}}\label{app:nike}
$\mathcal{B}$ proceeds through three games. $G_0$ is the real NIKE-slot experiment. $G_1$ aborts on any
forgery over an X25519 public value not generated by $\mathcal{B}$; the gap is $\eps_{\mathsf{sig}}$,
and $\mathcal{B}$ answers honest signing queries itself since it holds both ring secrets, generating
the client's signature as \HDKEM actually does (only the client signs). $G_2$ embeds the
Gap-DH challenge $(g^a,g^b)$ as the two X25519 public values and answers decision queries with the gap
oracle; at the test session $k_2=g^{ab}$ is the output. Summing gives
$\eps_{\mathsf{dh}}+\eps_{\mathsf{sig}}$. No step invokes anonymity or deniability of the ring
signature, which is why Theorem~\ref{thm:nike} is independent of Theorem~\ref{thm:deny}.

\section{Split-Key PRF Instantiation}\label{app:prf}
Instantiating $F$ as SHAKE256 with $64$-byte output and the injective length-prefixed encoding
of~\eqref{eq:skprf} meets Definition~\ref{def:skprf}, whose context-freshness restriction is discharged
by the per-session ML-KEM ciphertext in $\tau$. Two analyses bound the resulting advantage.

\emph{Monolithic model.} Modeling $F$ as a single quantum random oracle on the encoded input puts
Lemma~\ref{lem:skprf} in force, using~\cite[Thm.~1]{HRS16}. For $\ell\le 2^{40}$ sessions and
$q\le 2^{60}$ queries with $\beta=256$, $\eps_{\mathsf{prf}}\le 8\cdot 2^{120}\cdot 2^{-256}
\le 2^{-133}$.

\emph{Multi-oracle model.} Don, Fehr and Huang analyze the identical two-stage construction and obtain
the square-root bound of their Corollary~1. At $q_{\mathcal{H}} = 2^{60}$, $q_{\mathcal{O}} = 2^{40}$
and $\eps = 2^{-256}$ the two terms are $2^{-55.5}$ and $2^{-45.5}$; the second dominates. Raising the
source to $\beta = 384$ gives $2^{-109.5}$, which restores a PRF term below ML-KEM's $2^{-98}$.

For reference, the corresponding classical random-oracle bound would read $2^{-196}$ at $\beta=256$ and
$2^{-68}$ at $\beta=128$.The gap between that $2^{-68}$ and the correct quantum figures---approximately
$2^{-5}$ under the monolithic model and $2^{18.5}$, i.e.\ vacuous,
under the multi-oracle bound---is why $128$ bits are insufficient under
either analysis. All security theorems of \S\ref{sec:sec} hold under either
interpretation; only the numeric ceiling of \S\ref{sec:eck} changes.

\section{CK$^+$ QKD-Oracle Composition}\label{app:qkd}
\begin{theorem}[QKD-oracle source composition]\label{thm:qkdck}
In the CK$^+$ QKD-oracle model of H\"ovelmanns et al.~\cite{Hovelmanns25}, a fresh HDKEM test session whose QKD entry is neither leaked nor overridden has the eCK bound of Theorem~\ref{thm:eck} with $\eps_{\mathsf{src}}=0$ for the ideal QKD slot. If a concrete QKD service realizes that ideal key-delivery interface within trace distance $\eps_{\mathsf{qkd}}$, the bound increases additively by $\eps_{\mathsf{qkd}}$.
\end{theorem}

H\"ovelmanns et al.~\cite{Hovelmanns25} model a QKD service by two key-delivery calls: one endpoint obtains a fresh uniform key together with a unique public identifier, and the partner retrieves the matching key by that identifier. Their adversarial interface additionally permits leakage and replacement/override of QKD entries, which is necessary to expose dependent-key attacks. HDKEM maps these calls to $\mathsf{Src.Get}$ and $\mathsf{Src.GetWithID}$ and carries $\kid_3$ in $\rho_2$.

For the test session, cleanliness excludes leakage or override of the referenced QKD entry. Because $\kid_3$ is part of $\tau$, changing it after the client signature requires a ring-signature forgery. Replaying a valid transcript with an identifier that resolves to a different source key gives the two endpoints different $\kconf$ values and is rejected by Finished except with the AEAD/PRF terms already charged in Theorem~\ref{thm:eck}. Thus, conditioned on the usual authentication and confirmation games, the QKD slot used by both matching sessions is exactly the ideal oracle's uniform key, establishing Theorem~\ref{thm:qkdck}.

The theorem deliberately inherits the scope of the 2025 QKD-oracle model. First, its oracle returns an ideal uniform key; when a concrete QKD implementation is only $\eps_{\mathsf{qkd}}$-close to that key-delivery functionality, HDKEM's bound gains $\eps_{\mathsf{qkd}}$ under that realization assumption. Second, the CK$^+$ game is not a universal-concurrent-composition theorem: cross-protocol use of the same QKD service can require a stronger composable model, though Remark~\ref{rem:qkddomain} gives a deployment-level mitigation. Third, trusted-repeater QKD networks require additional trust assumptions, and CK$^+$ does not provide post-compromise security. None of these limitations re-opens the key-ID/dependent-key issue addressed by the protocol binding above.

\begin{remark}[Domain separation for shared QKD services]\label{rem:qkddomain}
If several protocols consume keys from the same QKD service, deployments
should domain-separate their key namespaces, for example by associating each
protocol with a distinct public protocol identifier and rejecting identifiers
presented in the wrong namespace. Such separation prevents direct
cross-protocol key-ID substitution: an identifier from another namespace is
either rejected by the source interface or changes the authenticated
transcript.

This is a deployment-level separation measure, not a general concurrent
composition theorem. Shared service state, correlated allocation, and other
cross-protocol interactions remain outside the CK$^+$ game considered here.
\end{remark}

\section{Source-less Fallback and Larger Rings}\label{app:fallback}
\paragraph{Source-less fallback:}
Without $k_3$, $S$ signs Step 1 under $\rsk_S$ to restore mutual explicit authentication. The protocol
retains OR-confidentiality over $k_1$ and $k_2$ (Thm.~\ref{thm:comb} with the third term dropped) and
resistance to active man-in-the-middle, but is eCK-secure only in C1 (\S\ref{sec:eckcases}).
Deniability still rests on the ring signature, with the two-signature bound of
Remark~\ref{rem:srcless}. The mode is bound into $\tau$. Where no pairwise secret is available,
\S\ref{sec:staticcmp} recommends a static-ephemeral design over this fallback.

\paragraph{Larger rings:}
For $|R|=m>2$, signing time grows linearly in $m$ and $\Auth$ grows with the number of commissioned
members, per Theorem~\ref{thm:alloc}(i) applied to each. All theorems hold unchanged, since they invoke
ring signature security generically; only the concrete $(\mu,\delta)$ and cost figures change.

\end{document}